\documentclass[a4paper,onecolumn,11pt,accepted=2021-05-02]{quantumarticle}
\pdfoutput=1
\usepackage[utf8]{inputenc}
\usepackage[english]{babel}
\usepackage[T1]{fontenc}
\usepackage{amsmath}
\usepackage{amsfonts}
\usepackage{amssymb}
\usepackage{hyperref}
\usepackage{tikz}
\usepackage{lipsum}

\usepackage{extarrows}
\usepackage{braket}
\usepackage{xcolor}
\usepackage{caption}
\usepackage{dirtytalk}
\usepackage{graphicx}
\usepackage{float}
\usepackage{authblk}
\usepackage{enumitem}

\newcommand{\E}{\mathcal{E}}
\newcommand{\T}{\mathcal{T}}
\newcommand{\Td}{\mathcal{T}^{ideal}_{\delta}}

\newcommand{\C}{\mathcal{C}}
\newcommand{\A}{\mathcal{A}}

\newcommand{\Gn}[3]{\mathcal{G}^{#1}_{#2, #3}}
\newcommand{\Gnn}[2]{\mathcal{G}^{#1}_{#2}}

\newcommand{\qEx}{\mathsf{qEx}}

\newcommand{\qSel}{\mathsf{qSel}}
\newcommand{\Sin}{S_{in}}
\newcommand{\Sout}{S_{out}}
\newcommand{\negl}{negl}
\newcommand{\nonnegl}{non\text{-}\negl}
\newcommand{\Hil}{\mathcal{H}}
\newcommand{\Hilin}{\mathcal{H}_{in}}
\newcommand{\Hilout}{\mathcal{H}_{out}}

\newcommand{\HilD}{\mathcal{H}^D}
\newcommand{\Hild}{\mathcal{H}^d}
\newcommand{\Hildperp}{\mathcal{H}^{d^{\perp}}}
\newcommand{\Hildperpo}{\mathcal{H}^{d^{\perp}}_{out}}

\newcommand{\Hildin}{\Hil^{d}_{in}}
\newcommand{\HildIn}{\Hil^{d_{in}}}
\newcommand{\Hildout}{\Hil^{d}_{out}}
\newcommand{\HildOut}{\Hil^{d_{out}}}

\newcommand{\kpo}{\ket{\psi^{out}}}
\newcommand{\po}{\psi^{out}}

\newcommand{\U}{\mathrm{U}}

\newcommand{\UqPUF}{\mathrm{UqPUF}}
\newcommand{\qPUF}{\mathrm{qPUF}}
\newcommand{\UqPUFid}{\UqPUF_{\mathbf{id}}}

\newcommand{\qPUFGen}{\mathrm{QGen}}
\newcommand{\qPUFEval}{\mathrm{QEval}}
\newcommand{\id}{\mathbf{id}}
\newcommand{\qPUFidi}{\qPUF_{\id_i}}
\newcommand{\qPUFidj}{\qPUF_{\id_j}}
\newcommand{\qPUFid}{\qPUF_{\id}}
\newcommand{\rhoin}{\rho_{in}}
\newcommand{\rhoout}{\rho_{out}}
\newcommand{\sigmain}{\sigma_{in}}
\newcommand{\sigmaout}{\sigma_{out}}
\newcommand{\psiin}{\psi_{in}}
\newcommand{\psiout}{\psi_{out}}

\newcommand{\mbraket}[2]{\bra{#1}#2\rangle}

\newtheorem{theorem}{Theorem}
\newtheorem{definition}{Definition}
\newtheorem{lemma}{Lemma}
\newenvironment{proof}{\textit{Proof:}}{\hfill$\square$}
\newenvironment{proofsketch}{\textit{Proof (Sketch):}}{\hfill$\square$}
\newtheorem{requirement}{Requirement}
\newtheorem{game}{Game}
\begin{document}

\title{Quantum Physical Unclonable Functions: Possibilities and Impossibilities}

\author{Myrto Arapinis}
\affiliation{School of Informatics, University of Edinburgh, 10 Crichton Street, Edinburgh EH8 9AB, UK}
\author{Mahshid Delavar}
\affiliation{School of Informatics, University of Edinburgh, 10 Crichton Street, Edinburgh EH8 9AB, UK}
\author{Mina Doosti}
\affiliation{School of Informatics, University of Edinburgh, 10 Crichton Street, Edinburgh EH8 9AB, UK}
\email{m.doosti@sms.ed.ac.uk}
\thanks{This work has been presented at QCrypt 2019 (9th International Conference on Quantum Cryptography)}
\orcid{0000-0003-0920-335X}
\author{Elham Kashefi}
\affiliation{School of Informatics, University of Edinburgh, 10 Crichton Street, Edinburgh EH8 9AB, UK}
\affiliation{Departement Informatique et Reseaux, CNRS, Sorbonne Universit\'{e}, 4 Place Jussieu 75252 Paris CEDEX 05, France}
\maketitle

\begin{abstract}
    A Physical Unclonable Function (PUF) is a device with unique behaviour that is hard to clone hence providing a secure fingerprint. A variety of PUF structures and PUF-based applications have been explored theoretically as well as being implemented in practical settings. Recently, the inherent unclonability of quantum states has been exploited to derive the quantum analogue of PUF as well as new proposals for the implementation of PUF. We present the first comprehensive study of quantum Physical Unclonable Functions (qPUFs) with quantum cryptographic tools. We formally define qPUFs, encapsulating all requirements of classical PUFs as well as introducing a new testability feature inherent to the quantum setting only. We use a quantum game-based framework to define different levels of security for qPUFs: quantum exponential unforgeability, quantum existential unforgeability and quantum selective unforgeability. 
    We introduce a new quantum attack technique based on the universal quantum emulator algorithm of Marvin and Lloyd to prove no qPUF can provide quantum existential unforgeability. On the other hand, we prove that a large family of qPUFs (called unitary PUFs) can provide quantum selective unforgeability which is the desired level of security for most PUF-based applications.
\end{abstract}

\section{Introduction}\label{sec:intro}
Canetti and Fischlin's result on the impossibility of achieving secure cryptographic protocols without any setup assumptions \cite{canetti2001universally} has motivated a rich line of research investigating the advantages of making hardware assumptions in protocol design. The idea was first introduced by Katz in \cite{katz2007universally}, and attracted the attention of researchers and developers as it adopts physical assumptions and eliminates the need to trust a designated party or to rely on computational assumptions. Among different hardware assumptions, Physical Unclonable Functions (PUFs) have greatly impacted the field \cite{badrinarayanan2017unconditional}.

PUFs are hardware structures designed to utilize the random physical disorder which appear in any physical device during the manufacturing process. Because of the uncontrollable nature of these random disorders, building a clone of the device is considered impractical. The behaviour of a PUF is usually equivalent to a set of Challenge-Response Pairs (CRPs) which are extracted through physically querying the PUF and measuring its responses. The PUF's responses depend on its physical features and are assumed to be unpredictable, i.e. even the manufacturer of the PUF, with access to many CRPs, cannot predict the response to a new challenge \cite{ruhrmair2014pufs}. This property makes PUFs different from other hardware tokens in the sense that the manufacturer of a hardware token is completely aware of the behaviour of the token they have built \cite{brzuska2011physically}. 

So far, the cryptographic literature has mainly considered what we will call classical PUFs (or cPUFs) restricted to classical CRPs. Most cPUFs generate only a finite, albeit possibly exponential (in some desired security parameters), number of CRPs \cite{chang2017retrospective}. However, most of them remain vulnerable against different attacks like side-channel \cite{tebelmann2019side,chang2017retrospective} and machine-learning \cite{ganji2016strong,ruhrmair2014puf,ruhrmair2010modeling,khalafalla2019pufs}. Thus, considering the importance of cPUFs as a hardware security primitive in several real-world applications, on one hand, \cite{chang2017retrospective,herder2014physical,delavar2017puf,ameri2019provably,marchand2017implementation,liu2019xor,mukhopadhyay2016pufs}\footnote{Recently SAMSUNG announced that in their new processor Exynox 9820 they have integrated SRAM based PUF to store and manage personal data in perfect isolation. Also, a UK company, Quantum Base, has started to mass-produce its patented optical quantum PUFs.} and the recent advances in quantum technology, on the other hand, it is worth investigating whether quantum technologies could boost the security of cPUFs or if they, on the contrary, threaten their security. In the current work, we address the general and formal treatment of PUFs in a quantum world for the first time by defining quantum PUFs (qPUFs) as a quantum token that can be challenged with quantum states and respond with quantum states. We identify the requirements a qPUF needs to meet to provide the main security property required for most of the qPUF-based applications, that is \emph{unforgeability}\footnote{\emph{Unpredictability} and \emph{unclonability} are other equivalent terms for this notion used often in the literature.}. All prior similar works~\cite{vskoric2010quantum,vskoric2012quantum,nikolopoulos2017continuous,young2019quantum} (see related work paragraph below) considered the special case of qPUFs where the encoding of the responses is known to the manufacturer and in fact, the evaluation of the qPUF is public information. 
We provide a general and formal mathematical framework for the study of qPUFs as a new quantum primitive inspired from the theoretical literature of classical PUF while taking into account full capabilities of a quantum adversary. However, it is worth mentioning that designing and implementing concrete qPUFs satisfying our proposed level of security set up remains a challenging task that we are exploring separately as a follow up of this work.

\subsubsection{Our Contributions.} We first define qPUFs as quantum channels and formalize the standard requirements of robustness, uniqueness and collision-resistance for qPUFs guided by the classical counterparts to establish the requirements that qPUFs should satisfy to enable their usage as a cryptographic primitive. We then use the game-based framework to define three security notions for qPUFs: quantum exponential unforgeability, quantum existential unforgeability and quantum selective unforgeability capturing the strongest type of attack models where the adversary has access to the qPUF and can query it with his chosen quantum states. In this new model, we demonstrate how quantum learning techniques, such as the universal quantum emulator algorithm of \cite{marvian2016universal}, can lead to successful attacks. In doing so we establish several possibility and impossibility results. 

\begin{itemize}
\item \emph{No qPUF provides Quantum Exponential Unforgeability.} 
The presented attack is the correct analogue of the brute-force attack for classical PUFs.
\item \emph{No qPUF provides Quantum Existential Unforgeability.} We show how the universal quantum emulator algorithm (which is polynomial in the size of the qPUF's dimension) can break this security property of any qPUFs. 
\item \emph{Any qPUF provides Quantum Selective Unforgeability.} In other words, no QPT adversary can, on average, generate the response of a qPUF to random challenges. 
\end{itemize}

\subsubsection{Other Related Works} The concept of Physical Unclonable Functions was first introduced by Pappu \emph{et al.} \cite{pappu2002physical} in 2001, devising the first implementation of an Optical PUF. Optical PUFs were subsequently improved as to generating an independent number of CRPs~\cite{mesaritakis2018physical}. Several structures of Physical Unclonable Functions were further introduced including Arbiter PUFs \cite{gassend2002silicon}, Ring-Oscillator based PUFs \cite{suh2007physical,delavar2016ring} and SRAM PUFs \cite{guajardo2007fpga}. For a comprehensive overview of existing PUF structures, we refer the reader to \cite{maes2016physically,halak2018physically}. 

Recently, the concept of \say{quantum read-out of PUF (QR-PUF)} was introduced in \cite{vskoric2010quantum} to exploit the no-cloning feature of quantum states to potentially solve the spoofing problem in the remote device identification. The QR-PUF-based identification protocol has been implemented in \cite{goorden2014quantum}. In addition to the security analysis of this protocol against intercept-resend attack in \cite{vskoric2010quantum}, its security has also been analysed against other special types of attacks targeting extracting information from an unknown challenge state \cite{vskoric2013security,yao2016quantum}. In another work, \cite{nikolopoulos2017continuous}, the continuous variable encoding is exploited to implement another practical QR-PUF based identification protocol. The security of this protocol has also been analysed only against an attacker who aims to efficiently estimate or clone an unknown challenge quantum state \cite{nikolopoulos2018continuous,fladung2019intercept}. Moreover, some other applications of QR-PUFs have been introduced in \cite{vskoric2017authenticated} and \cite{uppu2019asymmetric}.

In another independent recent work, Gianfelici et al. have presented a common theoretical framework for both cPUFs and QR-PUFs \cite{gianfelici2020theoretical}. They quantitatively characterize the PUF properties, particularly robustness and unclonability. They also introduce a generic PUF-based identification scheme and parameterize its security based on the values obtained from the experimental implementation of PUF.

\section{Quantum Emulation Algorithm}\label{sec:qe}
In this section, we describe the Quantum Emulation (QE) algorithm presented in \cite{marvian2016universal} as a quantum process learning tool that can outperform the existing approaches based on quantum tomography \cite{d2001quantum}. The main idea behind quantum emulation comes from the question on the possibility of emulating the action of an unknown unitary transformation on an unknown input quantum state by having some of the input-output samples of the unitary. An emulator is not trying to completely recreate the transformation or simulate the same dynamics. Instead, it outputs the action of the transformation on a quantum state. The original algorithm was developed and proposed in the context of quantum process tomography, thus the analysis did not consider any adversarial behaviour. For our cryptanalysis purposes, we need to provide a new fidelity analysis for challenges not fully lying within the subspace of the learning phase. We further optimise the success probability of our attack by optimising the choice of the reference state.

\subsection{The Circuit and Description}
The circuit of the quantum emulation algorithm is depicted in~Figure~\ref{fig1-qe} also in \cite{marvian2016universal} and works as follows: Let $\U$ be a unitary transformation on a D-dimensional Hilbert space $\HilD$, $\Sin = \{\ket{\phi_i}; i = 1, ..., K\}$ be a sample of input states and $\Sout = \{\ket{\phi^{out}_i}; i = 1, ..., K\}$ the set of corresponding outputs, i.e $\ket{\phi^{out}_i} = \U\ket{\phi_i}$. Also, let $d$ be the dimension of the Hilbert space $\Hild$ spanned by $\Sin$ and $\ket{\psi}$, a challenge state. The goal of the algorithm is to find the output of $\U$ on $\ket{\psi}$, that is $\U\ket{\psi}$.

\begin{figure*}[ht]
    \centering
\includegraphics[width=1.0\columnwidth]{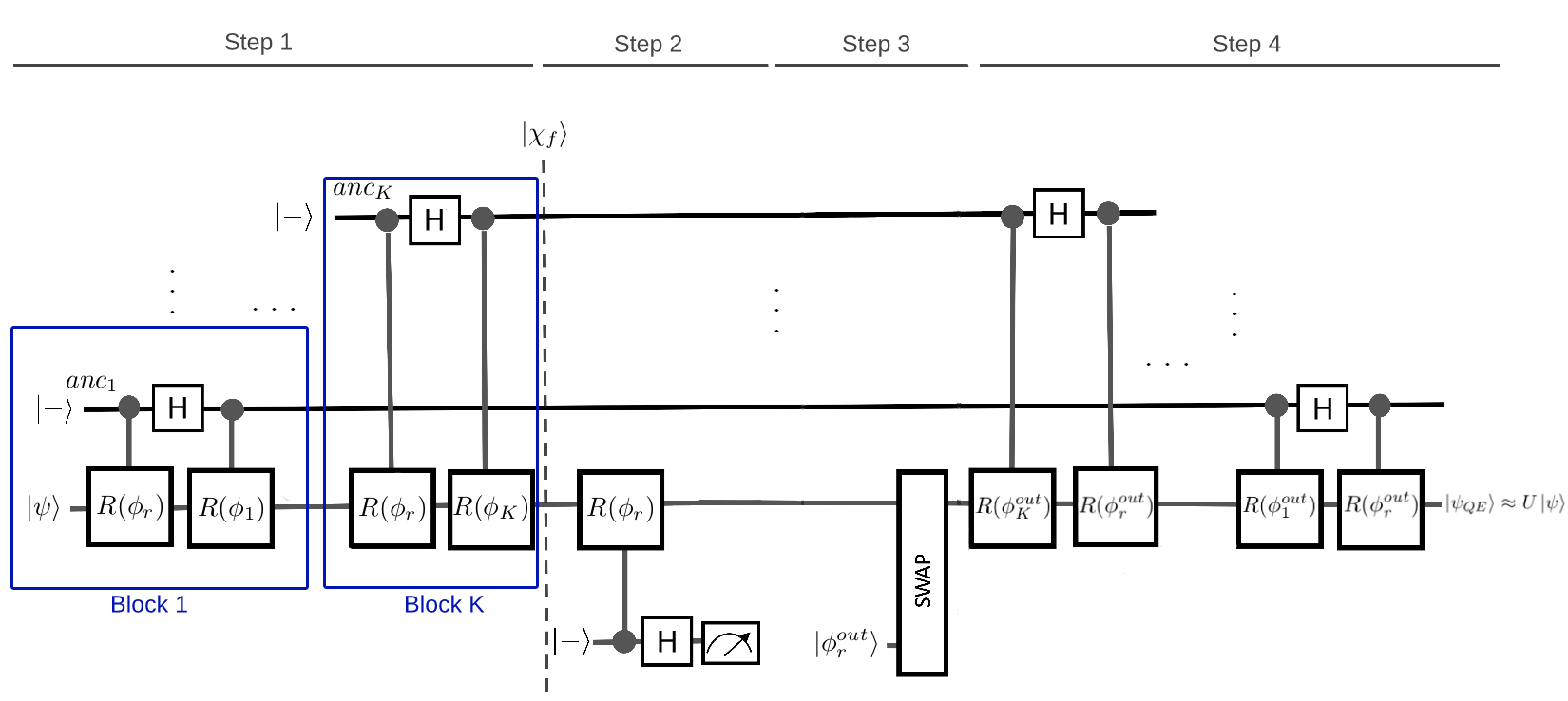}
    \caption{
    The circuit of the quantum emulation algorithm. $\ket{\phi_r}$ is the reference state and $\ket{\phi^{out}_r}$ is the output of the reference state. $R(*)$ gates are controlled-reflection gates. In each block of Step 1, a reflection around the reference and another sample state is being performed.
    }
    \label{fig1-qe}
\end{figure*}

The main building blocks of the algorithm are controlled-reflection gates described as:
\begin{equation}
    R_c(\phi) = \ket{0}\bra{0} \otimes \mathbb{I} + \ket{1}\bra{1} \otimes e^{i\pi\ket{\phi}\bra{\phi}}
\end{equation}
A controlled-reflection gate acts as the identity ($\mathbb{I}$) if the control qubit is $\ket{0}$, and as $R(\phi) = e^{i\pi\ket{\phi}\bra{\phi}} = \mathbb{I} - 2\ket{\phi}\bra{\phi}$ if the control qubit is $\ket{1}$.
The circuit also uses Hadamard and SWAP gates and consists of four stages.

\noindent\textbf{Stage 1.} $K$ number of sample states and a specific number of ancillary qubits are chosen and used through the algorithm. We assume the algorithm uses all of the states in $\Sin$. The ancillary systems are all qubits prepared at $\ket{-}$. Let $\ket{\phi_r} \in \Sin$ be considered as the reference state. This state can be chosen at random or according to a special distribution. The first step consists of $K-1$ blocks wherein each block, the following gates run on the state of the system and an ancilla:
    \begin{equation}\label{eq:qe-block}
        W(i) = R_c(\phi_i) H R_c(\phi_r).
    \end{equation}
    In each block represented by equation~(\ref{eq:qe-block}), a controlled-reflection around the reference state $\ket{\phi_r}$ is performed on $\ket{\psi}$ with the control qubit being on the $\ket{-}$ ancillary state. Then a Hadamard gate (H) runs on the ancilla followed by another controlled-reflection around the sample state $\ket{\phi_i}$. This is repeated for each of the $K$ states in $\Sin$ such that the input state is being entangled with the ancillas and also it is being projected into the subspace $\Hild$ in a way that the information of $\ket{\psi}$ is encoded in the coefficients of the general entangled state. This information is the overlap of $\ket{\psi}$ with all the sample inputs. By reflecting around the reference state in each block, the main state is pushed to $\ket{\phi_r}$ and the probability of finding the system at the reference state increases. The overall state of the circuit after Stage 1 is:
    \begin{equation}
        [W(K)...W(1)]\ket{\psi}\ket{-}^{\otimes K} \approx \ket{\phi_r}\ket{\Omega(anc)}
    \end{equation}
    where $\ket{\Omega(anc)}$ is the entangled state of $K$ ancillary qubits. The approximation comes from the fact that the state is not only projected on the reference quantum state but it is also projected on other sample quantum states with some probability. We present a more precise formula in the next subsection.\\
    
\noindent\textbf{Stage 2.} In this stage, first a reflection around $\ket{\phi_r}$ is performed and after applying a Hadamard gate on an extra ancilla, that ancilla is measured in the computational basis $\{\ket{0}, \ket{1}\}$. Based on the output of the measurement, one can decide whether the first step was successful (i.e. the output of the measurement is 0) or not. If the first step is successful, the main state has been pushed to the reference state. In this case, the algorithm proceeds with Stage 3. If the output is 1, the projection was unsuccessful and the input state remains almost unchanged. In this case, either the algorithm aborts or it goes back to the first stage and picks a new state as the reference. This stage has a post-selection role which can be skipped to output a mixed state of two possible outputs.\\
    
\noindent\textbf{Stage 3.} The main state is swapped with $\ket{\phi^{out}_r} = \U\ket{\phi_r}$ that is the output of the reference state. This is done by means of a SWAP gate. At this point, the overall state of the system is:
\begin{equation}
(\mathrm{SWAP}\otimes I^{\otimes K}) \ket{\phi^{out}_r}\ket{\phi_r}\ket{\Omega(anc)} = \ket{\phi_r}\ket{\phi^{out}_r}\ket{\Omega(anc)}.
\end{equation}
By tracing out the first qubit, the state of the system becomes $\ket{\phi^{out}_r}\ket{\Omega(anc)}$.\\
    
\noindent\textbf{Stage 4.} The last stage is very similar to the first one except that all blocks are run in reverse order and the reflection gates are made from corresponding output quantum states. The action of stage 4 is equivalent to:
\begin{equation}
W^{out}(i) = R_c(\phi^{out}_i) H R_c(\phi^{out}_r) = (\U \otimes I)W(i)(\U^{\dagger} \otimes \mathbb{I}).
\end{equation}
After repeating this gate for all the output samples, $\U$ is acted on the projected components of $\ket{\psi}$ and by restoring back the information of $\ket{\psi}$ from the ancilla, the input state approaches $\U\ket{\psi}$. The overall output state of the circuit at the end of this stage is:
\begin{equation}
[W^{out}(1)...W^{out}(K)]\ket{\phi^{out}_r}\ket{\Omega(anc)} \approx \U\ket{\psi}\ket{-}^{\otimes K}
\end{equation}
where equality is obtained whenever the success probability of Stage 2 is equal to 1. 

\subsection{Output fidelity analysis}
We are interested in the fidelity of the output state $\ket{\psi_{QE}}$ of the algorithm and the intended output $\U\ket{\psi}$ to estimate the success. In the original paper, the fidelity analysis is first provided for ideal controlled-reflection gates and later a protocol is presented to implement them efficiently. In this paper, as we are more interested in the theoretical bounds for the fidelity, all the gates including the controlled-reflection gates are assumed to be ideal keeping in mind that the implementation is possible \cite{marvian2016universal,lloyd2014quantum}. We recall the main theorem of~\cite{marvian2016universal}: 

\begin{theorem}\label{th:qe-fidel}\textbf{\cite{marvian2016universal}}
Let $\E_{\U}$ be the quantum channel that describes the overall effect of the algorithm presented above. Then for any input state $\rho$, the Uhlmann fidelity of $\E_{\U}(\rho)$ and the desired state $\U\rho\U^{\dagger}$ satisfies:
\begin{equation}\label{eq:qe-fidel}
    F(\rho_{QE}, \U\rho\U^{\dagger}) \geq F(\E_{\U}(\rho), \U\rho\U^{\dagger}) \geq \sqrt{P_{succ-stage1}}
\end{equation}
where $\rho_{QE} = \ket{\psi_{QE}}\bra{\psi_{QE}}$ is the main output state(tracing out the ancillas) when the post-selection in Stage 2 has been performed. $\E_{\U}(\rho)$ is the output of the whole circuit without the post-selection measurement in Stage 2 and $P_{succ-stage1}$ is the success probability of Stage 1.
\end{theorem}
For the purpose of this paper, we need a more precise and concrete expression for the output fidelity not covered in \cite{marvian2016universal}. From the proof of Theorem~\ref{th:qe-fidel} in \cite{marvian2016universal}, it can be seen that the success probability of Stage 1 is calculated as follows:
\begin{equation}\label{eq:qe-ps}
P_{succ-stage1} = |\bra{\phi_r}Tr_{anc}(\ket{\chi_f}\bra{\chi_f})\ket{\phi_r}|^2
\end{equation}
where $\ket{\chi_f}$ is the final state of the circuit after Stage 1 and $Tr_{anc}(\cdot)$ computes the reduced density matrix by tracing out the ancillas. The overlap of the resulting state and the reference state equals the success probability of Stage 1. Now relying on Theorem~\ref{th:qe-fidel}, we only use equation~(\ref{eq:qe-ps}) for our analysis henceforward.

The fidelity of the output state of the circuit highly depends on the choice of the reference state (equation~(\ref{eq:qe-ps})) such that it may increase or decrease the success probability of the adversary in different security models as we will discuss in the Section \ref{sec:qPUF}. We establish the following recursive relation for the state of the circuit after the $i$-th block of Stage 1, in terms of the previous state:
\begin{equation} \label{eq:qe-recur}
        \ket{\chi_i} = \frac{1}{2}[(I - R(\phi_r))\ket{\chi_{i-1}}\ket{0} + R(\phi_i)(\mathbb{I} + R(\phi_r))\ket{\chi_{i-1}}\ket{1}].
\end{equation}

Now by using this relation, we can prove the following theorem. The proof can be found in Appendix~\ref{app:qe-final-state-proof}

\begin{theorem}\label{th:qe-fins}
Let $\ket{\chi_K}$ be the output state of $K$-th block of the circuit  (Figure~\ref{fig1-qe}). Let $\ket{\psi}$ be the input state of the circuit, $\ket{\phi_r}$ the reference state and $\ket{\phi_i}$ other sample states. We have:
\begin{widetext}
\begin{equation}\label{eq:qe-fins}
\begin{split}
    \ket{\chi_K} &= \mbraket{\phi_r}{\psi} \ket{\phi_r}\ket{0}^{\otimes K} + \ket{\psi}\ket{1}^{\otimes K} - \mbraket{\phi_r}{\psi} \ket{\phi_r}\ket{1}^{\otimes K} \\
    & + \sum^{K}_{i=1}\sum^{i}_{j=0} [f_{ij} 2^{l_{ij}} |\mbraket{\phi_r}{\psi}|^{x_{ij}} |\mbraket{\phi_i}{\psi}|^{y_{ij}} |\mbraket{\phi_r}{\phi_i}|^{z_{ij}}]\ket{\phi_r}\ket{q_{anc}(i,j)} \\
    & + \sum^{K}_{i=1}\sum^{i}_{j=0} [g_{ij} 2^{l'_{ij}} |\mbraket{\phi_r}{\psi}|^{x'_{ij}} |\mbraket{\phi_i}{\psi}|^{y'_{ij}} |\mbraket{\phi_r}{\phi_i}|^{z'_{ij}}]\ket{\phi_i}\ket{q'_{anc}(i,j)}
\end{split}
\end{equation} 
\end{widetext}
where $l_{ij}$, $x_{ij}$, $y_{ij}$, $z_{ij}$, $l'_{ij}$, $x'_{ij}$, $y'_{ij}$ and $z'_{ij}$ are integer values indicating the power of the terms of the coefficient. Note that $f_{ij}$ and $g_{ij}$ can be 0, 1 or -1 and $q_{anc}(i,j)$ and $q'_{anc}(i,j)$ output a computational basis of $K$ qubits (other than $\ket{0}^{\otimes K}$).
\end{theorem}

Having a precise expression for $\ket{\chi_f}$ from Theorem~\ref{th:qe-fins}, one can calculate $P_{succ-step1}$ of equation~(\ref{eq:qe-ps}) by tracing out all the ancillary systems from the density matrix of $\ket{\chi_f}\bra{\chi_f}$. Also, now it is clear that if $\ket\psi$ is orthogonal to the $\Hild$, the only term remaining in equation~(\ref{eq:qe-fins}) is $\ket{\psi}\ket{1}^{\otimes K}$. So, the input state remains unchanged after the first stage and $P_{succ-step1}=0$.

For states projected in the subspace spanned by $\Sin$, the overall channel describing the quantum emulation algorithm has always a fixed point inside the subspace \cite{marvian2016universal}. Hence, Stage 1 is successful with probability close to 1 by assuming the gates to be ideal. 

\section{Quantum Physical Unclonable Functions}\label{sec:qPUF}
We consider a set of quantum devices that have been created through the same manufacturing process.
These devices respond with a general quantum state when challenged with a quantum state. Similar to the classical setting (see Appendix A), we formalize the manufacturing process of qPUFs by defining a $\qPUFGen$ algorithm:
\begin{equation}\label{eq:qPUFGen}
    \qPUFid\leftarrow \qPUFGen(\lambda)
\end{equation}
where $\mathbf{id}$ is the identifier of $\qPUFid$ and $\lambda$ the security parameter. 

We also need to define the $\qPUFEval$ algorithm mapping any input quantum state $\rhoin\in\HildIn$ to an output quantum state $\rhoout\in\HildOut$ where $\HildIn$ and $\HildOut$ are the domain and range Hilbert spaces of $\qPUFid$, denoted as:
\begin{equation}\label{eq:qPUFEval}
       \rhoout \leftarrow \qPUFEval(\qPUFid,\rhoin).
\end{equation}
 
For now, we allow for the most general form of trace-preserving quantum maps, i.e. CPT maps for $\qPUFEval$. So, we have:
\begin{equation}
    \rhoout = \Lambda_{\id}(\rhoin)
\end{equation}

Apart from these common algorithms (that are analogue to the classical setting), we also require qPUFs to include an efficient test algorithm $\T$ as we will formally define in Definition \ref{def:test} to test the equality between two unknown quantum states. We will also need the concept of quantum state distinguishability, which can be defined with different quantum distance measures such as trace distance or fidelity. Here we use the fidelity-based definition as follows: Let $F(\cdot,\cdot)$ denote the fidelity, and $\mu$ and $\nu$ the distinguishability and indistinguishability threshold parameters respectively such that $0 \leq \mu, \nu \leq 1$. We say two quantum states $\rho$ and $\sigma$ are $\mu$-distinguishable if $0\le F(\rho,\sigma) \leq 1-\mu$ and $\nu$-indistinguishable if $\nu \le F(\rho,\sigma) \le 1$. Finally, we can define a Quantum Physical Unclonable Function as follows. 


\begin{definition}[Quantum Physical Unclonable Function]\label{def:qPUF}
Let $\lambda$ be the security parameter, and $\delta_r, \delta_u, \delta_c\in[0,1]$ the robustness, uniqueness and collision resistance thresholds. A ($\lambda,\delta_r,\delta_u,\delta_c$)-qPUF includes the algorithms: $\qPUFGen$, $\qPUFEval$ and $\T$ satisfying Requirements \ref{def:robust}, \ref{def:uniq}, and \ref{def:col-res} defined below:
\end{definition}
\begin{requirement}[$\delta_r$-Robustness]\label{def:robust} For any $\qPUFid$ generated through $\qPUFGen(\lambda)$ and evaluated using $\qPUFEval$ on any two input states $\rhoin$ and $\sigmain$ that are $\delta_r$-indistinguishable, the corresponding output quantum states $\rhoout$ and $\sigmaout$ are also $\delta_r$-indistinguishable with overwhelming probability,
\begin{equation}
 \mathrm{Pr}[\delta_r\le F(\rhoout, \sigmaout)\le 1] = 1-\negl(\lambda).
\end{equation}
\end{requirement}

\begin{requirement}[\textbf{$\delta_u$-Uniqueness}]\label{def:uniq} For any two $\qPUF$s generated by the $\qPUFGen$ algorithm, i.e. $\qPUFidi$ and $\qPUFidj$, the corresponding CPT map models, i.e. $\Lambda_{id_i}$ and $\Lambda_{id_j}$ are $\delta_u$-distinguishable with overwhelming probability, 
\begin{equation}
\mathrm{Pr}[\ \parallel (\Lambda_{id_i} - \Lambda_{id_j})_{i\neq j}\parallel_\diamond \ge \delta_{u} \ ]=1-\negl(\lambda).
\end{equation}
\end{requirement}

\begin{requirement}
[\textbf{$\delta_c$-Collision-Resistance (Strong)}]\label{def:col-res} For any $\qPUFid$ generated by $\qPUFGen(\lambda)$ and evaluated by $\qPUFEval$ on any two input states $\rhoin$ and $\sigmain$ that are $\delta_c$-distinguishable, the corresponding output states $\rhoout$ and $\sigmaout$ are also $\delta_c$-distinguishable with overwhelming probability,\footnote{A weaker variant of Collision-Resistance, with separate input/output bound can be also defined in a similar fashion where the responses generated by $\qPUFEval$ on any two $\delta^i_c$-distinguishable input states $\rhoin$ and $\sigmain$, should be at least $\delta^o_c$-distinguishable. In fact, if $\delta^i_c = \delta^o_c = \delta_c$ we call the requirement a strong collision-resistance. Note that this equality holds up to a negligible value in the security parameter, i.e. if $\delta^i_c = \delta^o_c \pm \negl(\lambda)$, the strong collision-resistance requirement has still been satisfied. If $\delta^o_c < \delta^i_c$ (the difference is non-negligible) then this is referred to as weak collision-resistance.}
\begin{equation}
\mathrm{Pr}[0\le F(\rhoout, \sigmaout)\le 1-\delta_c] = 1-\negl(\lambda).
\end{equation}
\end{requirement}

In qPUF-based applications such as device authentication (or identification), it is necessary that there be a clear distinction between different qPUF instances generated by the same $\qPUFGen$ algorithm running on the same parameters $\lambda$ \cite{armknecht2016towards}. To this end, the following conditions should be satisfied: $\delta_c \leq 1-\delta_r$ and $\delta_u \leq 1-\delta_r$. So, we can drop $\delta_u$ and $\delta_c$ from the notation and characterize the qPUF as $(\lambda,\delta_r)-\qPUF$.

We also need to mention that, $\delta_r$ and $\delta_c$ parameters can allow for some specific noise models for each PUF device. More specifically, the collision resistance parameter i.e. $\delta_c$ or the ratio of $\delta^o_c / \delta^i_c$ is directly related to the channel parameters of the qPUF evaluation. Although, as the collision-resistance is an important requirement for achieving a secure PUF, similar to classical PUFs, we choose the strong collision-resistance as the main requirement for the quantum PUF. We specify that the strong collision-resistance parameter can allow for noisy PUF evaluation under the coherent noise models. Such noise models preserve distances between the input and output states of the qPUF and this property makes them suitable candidates for quantum PUF. Also, it has been shown in~\cite{greenbaum2017modeling} that a general noise can be modelled as a combination of coherent and incoherent noises. Hence only the class of noise model with an almost close to zero incoherent factor can be considered to satisfy the $\delta_c$ (strong) collision resistance. Hence for the rest of this work, aiming to formalise the first general security framework, we consider a noiseless setting and leave further investigation that would be linked to particular construction to future works.

We have initially allowed for any CPT map as $\qPUFEval$ algorithm. Now, we let the $\qPUFEval$ algorithm be a CPT map with the same dimension of domain and range Hilbert space, i.e. $d_{in}=d_{out}$. We show that under this assumption, only unitary transformations and CPT maps that are negligibly close to unitary, can simultaneously provide the (strong)collision-resistance and robustness requirements of qPUFs. 

\begin{theorem}\label{theorem:non-unitary}
Let $\E(\rho)$ be a completely positive and trace-preserving (CPT) map described as follows:
\begin{equation}\label{eq:non-unitary-puf}
    \E(\rho) = (1-\epsilon)U \rho U^{\dagger} + \epsilon \Tilde{\E}(\rho) 
\end{equation}
where $U$ is a unitary transformation, $\Tilde{\E}$ is an arbitrary (non-negligibly) contractive channel and $0 \leq \epsilon \leq 1$. Then $\E(\rho)$ is a ($\lambda,\delta_r,\delta_c$)-qPUF for any $\lambda$, $\delta_r$, and $\delta_c$ and with the same dimension of domain and range Hilbert space, if and only if $\epsilon = \negl(\lambda)$.
\end{theorem}

\begin{proof}
First, we note that The contractive property of trace-preserving operations \cite{nielsen2010quantum} states that CPT maps on the same Hilbert space, can only preserve or decrease distances thus we have:
\begin{equation}
  F(\mathcal{E}(\rho), \mathcal{E}(\sigma)) \geq F(\rho, \sigma)
\end{equation}
Thus the robustness is generally satisfied. As a result, the proof of the theorem reduces to proving for collision-resistance. Let $\rho$ and $\delta$ be two $\delta_c$-distinguishable challenge with fidelity $F(\rho, \sigma) \leq 1-\delta_c$. Again with the above argument the fidelity of the outputs cannot be smaller than $F(\rho, \sigma)$. Thus the $\delta_c$ requirement is satisfied if the fidelity of the response density matrices are equal up to a negligible value.

Now let $\rho_1 = U\rho U^{\dagger}$, $\sigma_1 = U\sigma U^{\dagger}$, $\rho_2 = \Tilde{\E}(\rho)$, and $\sigma_2 = \Tilde{\E}(\sigma)$. We use the joint concavity of the fidelity~\cite{nielsen2010quantum} to obtain the following relation for the channel's output fidelity:
\begin{equation}\label{eq:fidelity-joint-concavity}
\begin{split}
        F(\E(\rho),\E(\sigma)) & = F((1-\epsilon)\rho_1 + \epsilon \rho_2, (1-\epsilon)\sigma_1 + \epsilon \sigma_2) \\
        & \geq (1-\epsilon)F(\rho_1,\sigma_1) + \epsilon F(\rho_2,\sigma_2)
\end{split}
\end{equation}
Since the first part of the channel is unitary which is distance preserving, we have $F(\rho_1,\sigma_1) = F(\rho,\sigma)$. Also due to contractive property of trace-preserving operations we know that $F(\rho_2,\sigma_2) \geq F(\rho,\sigma)$. We have
\begin{equation}
    F(\E(\rho),\E(\sigma)) - F(\rho,\sigma) \geq \epsilon (F(\rho_2,\sigma_2)-F(\rho,\sigma))
\end{equation}
Now since the channel $\Tilde{\E}$ is non-negligibly contractive, the value $F(\rho_2,\sigma_2)-F(\rho,\sigma)$ is not necessarily negligible and in order for the LHS of ~\ref{eq:fidelity-joint-concavity} to be always negligible, $\epsilon$ has to be negligible. So we have proved that CPT maps of the form~\ref{eq:non-unitary-puf} can be $\delta_c$ collision resistance qPUFs only if $\epsilon = \negl(\lambda)$.

Now we show that all channels of the form of Equation~\ref{eq:non-unitary-puf} where $\epsilon$ is negligible satisfy the strong collision resistance property up to a negligible value. To show that we use the relation between fidelity and trace distance which we denote as $\mathcal{D}_{tr}$, which is $\mathcal{D}_{tr}(\rho, \sigma) \leq \sqrt{1 - F(\rho,\sigma)}$. We use this inequality to relate the distance between the states $\E(\rho)$ and $\E(\sigma)$ and the original distance between $\rho$ and $\sigma$ and we subtract both sides to get the following inequality:
\begin{widetext}
\begin{equation}
\begin{split}
        F(\E(\rho),\E(\sigma)) - F(\rho,\sigma) & \leq \mathcal{D}^2_{tr}(\rho, \sigma) - \mathcal{D}^2_{tr}(\E(\rho),\E(\sigma))\\
        & \leq (\mathcal{D}_{tr}(\rho, \sigma) - \mathcal{D}_{tr}(\E(\rho),\E(\sigma)))(\mathcal{D}_{tr}(\rho, \sigma) + \mathcal{D}_{tr}(\E(\rho),\E(\sigma)))\\
        & \leq 2(\mathcal{D}_{tr}(\rho, \sigma) - \mathcal{D}_{tr}(\E(\rho),\E(\sigma)))
\end{split}
\end{equation}
\end{widetext}

In Appendix~\ref{ap:channel-output-ditance}, Lemma~\ref{lemma:epsilon-cpt-trace-distance} we show that the difference between the trace distance of the input and output for channels described as Equation~\ref{eq:non-unitary-puf}, is bounded by $\epsilon \mathcal{D}_{tr}(\rho, \sigma)$. Thus we have:
\begin{equation}
    F(\E(\rho),\E(\sigma)) - F(\rho,\sigma) \leq 2\epsilon \mathcal{D}_{tr}(\rho, \sigma)
\end{equation}
Now since $\epsilon = \negl(\lambda)$ and $0 \leq \mathcal{D}_{tr}(\rho, \sigma) \leq 1$, we can conclude that the difference between the fidelity is also negligible and hence the $\delta_c$ collision-resistance is satisfied up to a negligible value, and the proof is complete.
\end{proof}

The above theorem shows that only unitary or more generally, \emph{$\epsilon$-disturbed unitary maps} where $\epsilon$ is small, are suitable candidates for qPUF, especially when strong collision resistance is required. Thus, in the rest of the paper, we choose the $\qPUFEval$ algorithm to be a unitary map, and also for simplicity, we establish some of our theorems with pure quantum states, noting that considering the mixed states would not affect the main results. We call this type of qPUFs, Unitary qPUFs (or simply UqPUFs) and formally define them in Definition~\ref{def:uqPUF}. Nevertheless, we believe studying more general non-unitary qPUFs will be interesting future research directions in this field. 

Moreover, we require UqPUF transformations to be initially unknown (or exponentially hard to recover) as we will formally define in Definition~\ref{def:unk-uni}. This is a hardware assumption that is also considered in the classical setting where the PUF behaviour is unknown even for the manufacturer \cite{ruhrmair2014pufs}. Although from a construction point of view, this may not seem an easily achievable requirement, from a practical point of view this assumption is reasonable considering limited fabrication capabilities or the fact that simulating the same unitary on a quantum computer is not technologically easy due to noise or accumulated errors in each gate, even when the structure of the unitary is known. Moreover, there are promising constructions such as the family of optical schemes implemented using crystals or optical scattering media~\cite{nikolopoulos2017continuous}, where usually even the manufacturer does not know the underlying unitary unless querying it. On the other hand, in gate-based construction, one cannot avoid the fact that the manufacturer knows the underlying unitary. Hence this type of constructions cannot provide security against an adversarial manufacturer. Nevertheless, if predicting the evolution of a quantum state is difficult this is enough for security under the usual PUF assumptions. Hence such devices are still useful and practical for many applications as they can still provide security against any malicious adversary other than the manufacturer. We also note that from the theoretical point of view, this requirement is a minimal and pre-challenge requirement that can be achieved by sampling a family of unitaries indistinguishable from the Haar family of unitary transformations in single-shot, and we believe there are efficient ways to do this sampling~\cite{dankert2006c,ambainis2007quantum}. Finally, our framework and results cover both adversarial models where the manufacturer could be trusted or not.

\begin{definition}[Unknown Unitary Transformation]\label{def:unk-uni} We say a family of unitary transformations $U^u$, over a $D$-dimensional Hilbert space $\HilD$ is called Unknown Unitaries, if for all QPT adversaries $\A$ the probability of estimating the output of $U^u$ on any randomly picked state $\ket{\psi}\in\HilD$ is at most negligibly higher than the probability of estimating the output of a Haar random unitary operator on that state:
\begin{widetext}
\begin{equation}
        |\underset{U \leftarrow U^u}{Pr}[F(\A(\ket{\psi}),U\ket{\psi}) \geq \nonnegl(\lambda)] - \underset{U_{\mu} \leftarrow \mu}{Pr}[F(\A(\ket{\psi}),U_{\mu}\ket{\psi})
        \geq \nonnegl(\lambda)]| = \negl(\lambda).
\end{equation}
\end{widetext}
where $\mu$ denotes the Haar measure and the average probability has been taken over al; the states $\ket{\psi}$. 
\end{definition}
Note that UqPUFs also satisfy a natural notion of unclonability, known as no-cloning of unitary transformation \cite{chiribella2008optimal} which states that two black-box unitary transformations $\mathcal{O}_1$ and $\mathcal{O}_2$ cannot be perfectly cloned by a single-use apart from the trivial cases of perfect distinguishability or when $\mathcal{O}_1 = \mathcal{O}_2$. Thus, two UqPUFs, as long as they correspond to different unitaries, which is satisfied by the uniqueness requirement, are unclonable by quantum mechanics through a single-use. In the following section, we then show how this unclonability property can be extended to the case where the transformation has been used multiple times by formally introducing the notion of unforgeability. Thus, we define the unitary qPUFs as follows.

\begin{definition}[Unitary qPUF (UqPUF)]\label{def:uqPUF} A Unitary $\qPUF$ $((\lambda,\delta_r)-\UqPUF)$ is a $(\lambda,\delta_r)-\qPUF$ where the $\qPUFEval$ algorithm is modelled by an unknown unitary transformation
$\U_\id$ over a $D$-dimensional Hilbert space, $\HilD$ operating on pure input quantum states $\ket\psiin\in\HilD$ and returning pure output quantum states $\ket\psiout\in\HilD$,
\begin{equation}
    \ket{\psiout} = \qPUFEval(\UqPUFid, \ket{\psiin}) = \U_\id\ket{\psiin}.
\end{equation}
\end{definition}

As a result of the distance-preserving property of UqPUFs, we drop $\delta_r$ from the notation and simply characterise UqPUF as $\lambda$-UqPUFs.

\subsection{Security notion for qPUFs}\label{sec:UqPUF-game}
The security of most PUF-based applications such as PUF-based identification protocols relies on the unforgeability of PUFs \cite{armknecht2016towards}. Informally, unforgeability means that given a subset of challenge-response pairs of the target PUF, the probability of correctly guessing a new challenge-response pair shall be negligible in terms of the security parameter. In this section, we formally define this security notion for qPUFs in a game-based framework which is a standard framework for defining security of cryptographic primitives and analysing their security \cite{armknecht2016towards,boneh2013secure,delavar2017puf}. 

Accordingly, we define unforgeability as a game between an adversary who represents the malicious party and a challenger who plays the role of the honest party. The game is run in four steps: Setup, Learning, Challenge and Guess. 

In the \emph{setup phase}, the necessary public and private parameters and functions are shared between the adversary and the challenger. 

The \emph{learning phase} models the amount of knowledge that the adversary can get from the challenger. Similar to \cite{armknecht2016towards}, we consider chosen-input attacks modelling an adversary that has access to the qPUF and can query it with his own chosen inputs from the domain Hilbert space. Because of the quantum nature of the adversary's queries, the adversary has to prepare two copies of each query, keep one in his database and send the other one to the challenger. 

The \emph{challenge phase} captures the intended security notion. We consider here two types of challenge phase: Existential and Selective. In an existential challenge phase, the adversary chooses the challenge state while in a selective one, the challenge state is chosen by the challenger. We characterize a "new" existential challenge by imposing the adversary to choose a state that is $\mu$-distinguishable from all the inputs queries in the learning phase. In the selective case, to ensure the adversary has no knowledge about the challenge, we impose the challenger to choose the challenge uniformly at random from the domain Hilbert space. 

Finally, in the \emph{guess phase}, the adversary outputs his guess of the response corresponding to the challenge chosen in the challenge phase. The challenger checks the equality between the adversary's guess and the correct response with a test algorithm. The adversary wins the game if the output of the test algorithm is 1. Due to the impossibility of perfectly distinguishing all quantum states, checking equality of two completely unknown states is a non-trivial task. This is one of the major differences between classical and quantum PUFs. Nevertheless, a probabilistic comparison of unknown quantum states can be achieved through the simple quantum SWAP test algorithm~\cite{buhrman2001quantum}, and its generalisation to multiple copies introduced recently in~\cite{chabaud2018optimal}. Here we abstract from specific tests and define necessary conditions for a general quantum test.

\begin{definition}[Quantum Testing Algorithm]\label{def:test} Let $\rho^{\otimes \kappa_1}$ and $\sigma^{\otimes \kappa_2}$ be $\kappa_1$ and $\kappa_2$ copies of two quantum states $\rho$ and $\sigma$, respectively. A Quantum Testing algorithm $\T$ is a quantum algorithm that takes as input the tuple ($\rho^{\otimes \kappa_1}$,$\sigma^{\otimes \kappa_2}$) and accepts $\rho$ and $\sigma$ as equal (outputs 1) with the following probability 
\begin{equation*}
\small
    \mathrm{Pr}[1 \leftarrow \T(\rho^{\otimes \kappa_1}, \sigma^{\otimes \kappa_2})] = 1 - \mathrm{Pr}[0 \leftarrow \T(\rho^{\otimes \kappa_1}, \sigma^{\otimes \kappa_2})] = f(\kappa_1,\kappa_2, F(\rho, \sigma))
\end{equation*}
where $F(\rho, \sigma)$ is the fidelity of the two states and $f(\kappa_1,\kappa_2, F(\rho, \sigma))$ satisfies the following limits:
\begin{equation}
 \begin{cases}
    \lim_{F(\rho, \sigma) \rightarrow 1}f(\kappa_1,\kappa_2, F(\rho, \sigma)) = 1  \quad \forall\:(\kappa_1,\kappa_2)\\
    \lim_{\kappa_1,\kappa_2 \rightarrow \infty}f(\kappa_1,\kappa_2, F(\rho, \sigma)) = F(\rho, \sigma)\\
    \lim_{F(\rho, \sigma) \rightarrow 0}f(\kappa_1,\kappa_2, F(\rho, \sigma)) = Err(\kappa_1, \kappa_2)
  \end{cases} 
\end{equation}
with $Err(\kappa_1,\kappa_2)$ characterising the error of the test algorithm and $F(\rho, \sigma)$ the fidelity of the states.
\end{definition}

We also define another abstraction of the test algorithm in an ideal case which later helps us to demonstrate the security of the UqPUF. We formalize the ideal test $\Td$ as follows:

\begin{definition}[$\Td$ Test Algorithm]\label{def:test-delta} We call a test algorithm according to Definition~\ref{def:test}, a $\Td$ Test Algorithm when for any two state $\ket{\psi}$ and $\ket{\phi}$ the test responds as follows:
\begin{equation}
 \Td = \begin{cases}
    1  & \:F(\ket{\psi}, \ket{\phi})\geq \delta\\
    0  & \:otherwise
  \end{cases} 
\end{equation}
\end{definition}

Now we are ready to formalize unforgeability through a formal security game. 
\begin{game}[Formal game-based security of qPUF] Let $\qPUF=(\qPUFGen,\qPUFEval,\T)$ and $\T$ be defined as Definition \ref{def:qPUF} and \ref{def:test}, respectively. We define the following game $\Gn{\qPUF}{c}{\mu}(\A,\lambda)$ running between an adversary $\A$ and a challenger $\C$:
\begin{itemize}
    \item [] \textbf{Setup.} The challenger $\C$    
    runs $\qPUFGen(\lambda)$ to build an instance of the qPUF family, $\qPUFid$. Then, $\C$ reveals to the adversary $\A$, the domain and range Hilbert space of $\qPUFid$ respectively denoted by $\Hilin$ and $\Hilout$
    as well as the identifier of $\qPUFid$, $\mathbf{id}$. The challenger initialises two empty databases, $\Sin$ and $\Sout$ and shares them with the adversary $\A$. Also $\Hildin$ denotes adversary's input subspace.
    \item [] \textbf{Learning.} For $i=1:k$
        \begin{itemize}
            \item $\A$ prepares two copies of a quantum state $\rho_i \in \Hildin$, appends one to $\Sin$ and sends the other to $\C$;
            \item $\C$ runs $\qPUFEval(\qPUFid,\rho_i)$ and sends $\rho_i^{out}$, to $\A$;
            \item $\A$ appends $\rho_i^{out}$ to $\Sout$.
        \end{itemize} 
    \item [] \textbf{Challenge.}\footnote{The parameter $c$ specifies the type of the challenge phase.}
    \begin{itemize}
        \item If $c=\qEx$: $\A$ picks a quantum state $\rho^*\in\Hildin$ at least $\mu$-distinguishable from all the states in $\Sin$ and sends $\kappa_1$ copies of it to $\C$;
        \item If $c=\qSel$: $\C$ chooses a quantum state $\rho^*$ at random from the uniform distribution over the Hilbert space $\Hildin$. The challenger keeps $\kappa_1$ copies of $\rho^*$ and sends an extra copy of $\rho^*$ to $\A$.
    \end{itemize}
    \item [] \textbf{Guess.} 
    \begin{itemize}
        \item $\A$ sends $\kappa_2$ copies of his guess $\rho'$ to~ $\C$;
        \item $\C$ runs $\qPUFEval(\qPUFid,\rho^*)^{\otimes\kappa_1}$, and gets $\rho_{out}^{*\otimes\kappa_1}$;
        \item $\C$ runs the test algorithm $b\leftarrow \mathcal{T}(\rho_{out}^{*\otimes \kappa_1},\rho'^{\otimes \kappa_2})$
        where $b\in\{0,1\}$ and outputs $b$. The adversary wins the game if $b=1$.\footnote{Note that all the learning phase queries and the challenges represented with $\rho, \rho', \ket{\phi}$, etc. are considered to be any general separable or entangled state of a $D$-dimensional Hilbert space. Moreover, $\kappa_1$ and $\kappa_2$ are a choice of notation that enables us to include any desired quantum test algorithm according to Definition~\ref{def:test} and are independent of the number of the copies that the adversary uses in the learning phase.}
    \end{itemize}
\end{itemize}
\end{game}

Based on the above game, we define the security notions, \emph{quantum exponential unforgeability}, \emph{quantum existential unforgeability} and \emph{quantum selective unforgeability} for qPUFs; where the first one, models unforgeability of qPUFs against exponential adversaries with unlimited access to the qPUF in the learning phase; the second one is the most common and strongest type of unforgeability against Quantum Polynomial-Time (QPT) adversaries; finally the third one is a weaker notion of unforgeability that is sufficient for most qPUF-based applications like qPUF-based identification protocols. 

\begin{definition}[Quantum Exponential Unforgeability]\label{def:QunconUf}
A $\qPUF$ provides quantum exponential unforgeability if the success probability of any \emph{exponential} adversary $\A$ in winning the game $\Gn{\qPUF}{\qEx}{\mu}(\A, \lambda)$ is negligible in $\lambda$
\begin{equation}
    Pr[1\leftarrow \Gn{\qPUF}{\qEx}{\mu}(\A, \lambda)] = \negl(\lambda)    
\end{equation}
\end{definition}

\begin{definition}[$\mu$-Quantum Existential Unforgeability]\label{def:Qunf}
A $\qPUF$ provides $\mu$-quantum existential unforgeability if the success probability of any Quantum Polynomial-Time (QPT) adversary $\A$ in winning the game $\Gn{\qPUF}{\qEx}{\mu}(\A, \lambda)$ is negligible in $\lambda$
\begin{equation}
Pr[1\leftarrow \Gn{\qPUF}{\qEx}{\mu}(\A, \lambda)] = \negl(\lambda)
\end{equation}
\end{definition}

\begin{definition}[Quantum Selective Unforgeability]\label{def:Qsun}
A $\qPUF$ provides quantum selective unforgeability if the success probability of any Quantum Polynomial-Time (QPT) $\A$ in winning the game $\Gnn{\qPUF}{\qSel}(\lambda, \A)$ is negligible in $\lambda$
\begin{equation}
Pr[1\leftarrow \Gnn{\qPUF}{\qSel}(\lambda, \A)] = \negl(\lambda)
\end{equation}
\end{definition}

\subsection{Security analysis of Unitary qPUFs}

Here, we show which security notions defined in Section 4.1 can be achieved by unitary qPUFs (UqPUFs) over a $D$-dimensional Hilbert space operating on pure quantum states. 

In the classical setting, cPUFs can be fully described by the finite set of CRPs, and this suffices for breaking unforgeability. More precisely, an unbounded or exponential adversary can extract the entire set of CRPs by querying the target cPUF with all possible challenges \cite{chang2017retrospective}. If the challenges are n-bit strings, the number of possible challenges is $2^n$. However, in the quantum setting, a UqPUF can generate an infinite number of quantum challenge-response pairs such that extracting all of them is hard, even for exponential adversaries. This, combined with limitations imposed by quantum mechanics such as no-cloning \cite{wootters1982single} and the limits on state estimation \cite{bruss1998optimal}, raise the question if UqPUFs could satisfy unforgeability against exponential adversaries. 
We now prove that no UqPUF provides quantum exponential unforgeability as defined in Definition~\ref{def:QunconUf}.

\begin{theorem}\label{th:no-qPUF-unbound}
\textbf{(No UqPUF provides quantum exponential unforgeability)} For any $\lambda$-UqPUF and any $0\leq\mu\leq1$, there exists an exponential quantum adversary $\A$ such that
\begin{equation}
Pr[1\leftarrow \Gn{\UqPUF}{\qEx}{\mu}(\lambda, \A)] = \nonnegl(\lambda)    
\end{equation}
\end{theorem}
\begin{proof}
The key idea of the proof is based on complexity analysis of unitary tomography and implementation of a general unitary by single and double qubit gates, since for an exponential quantum adversary, it will be feasible to extract the unitary matrix by tomography and then build the extracted unitary by general gate decomposition method. By using the Solovay-Kitaev theorem \cite{nielsen2010quantum}, we then show that the adversary can build the unitary matrix of the UqPUF performing on $n$-qubits, within an arbitrarily small distance $\epsilon$ using $O(n^2 4^n \log^c(n^2 4^n))$ gates and hence win the game with any test algorithm $\T$.
Let $\UqPUF_\id$ operate on $n$-qubit input-output pairs where $n=\log(D)$. In the learning phase, $\A$ selects a complete set of orthonormal basis of $\HilD$ denoted as $\{\ket{b_i}\}^{2^n}_{i=1}$ and queries $\UqPUF_\id$ with each base $2^n$ times. So, the total number of queries in the learning phase is $k_1=2^{2n}$.

Then, $\A$ runs a \emph{unitary tomography} algorithm to extract the mathematical description of the unknown unitary transformation corresponding to the $\UqPUF_\id$, say $\U_\id$. It has been shown in \cite{nielsen2010quantum} that the complexity of this algorithm is $\mathcal{O}(2^{2n})$ for n-qubit input-output pairs. This is feasible for an exponential adversary. It is clear that once the mathematical description of the unitary is extracted, $\A$ can simply calculate the response of the unitary to a known challenge quantum state and wins the game $\Gn{\UqPUF}{\qEx}{\mu}(\lambda, \A)$ for any value of $\mu$. So, we have:
\begin{equation}
    Pr[1\leftarrow \Gn{\UqPUF}{\qEx}{\mu}(\lambda, \A)] = 1.
\end{equation}
We can also show the exponential adversary wins even the weaker notion of the security, i.e. quantum selective unforgeability, where he has only one copy of the challenge quantum state. 
To win the game with the selective challenge phase, the adversary needs to implement the unitary. 

It is known that any unitary transformation over $\Hil^{2^n}$ requires $O(2^{2n})$ two-level unitary operations or $O(n^2 2^{2n})$ single qubit and CNOT gates \cite{nielsen2010quantum} to be implemented. However, according to Solovay-Kitaev theorem \cite{nielsen2010quantum}, to implement a unitary with an accuracy $\epsilon$ using any circuit consisting of $m$ single qubit and CNOT gates, $O(m \log^c(m/c))$ gates from the discrete set are required where $c$ is a constant approximately equal to 2. Thus, an arbitrary unitary performing on $n$-qubit can be approximately implemented within an arbitrarily small distance $\epsilon$ using $O(n^2 4^n \log^c(n^2 4^n))$ gates.

So, $\A$ implements the unitary $\U'_\id$ with error $\epsilon$. Let $\A$ get the challenge state $\ket{\psi}$ in the $\qSel$ Challenge phase. The adversary queries $\U'_\id$ with $\ket\psi$ and gets $\ket\omega=\U'_\id\ket\psi$ as output. Since the $\epsilon$ can be arbitrary small, then $F(\U_\id\ket\psi,\U'_\id\ket\psi)\ge 1-\negl(\lambda)$. So, $\A$'s output $\ket\omega$ passes any test algorithm $\T(\kpo^{\otimes \kappa_1}, \ket{\omega}^{\otimes \kappa_2})$ with probability close to 1. Again, an unbounded adversary wins the game $\Gn{\UqPUF}{\qSel}{\mu}(\lambda, \A)$ with probability 1.
\end{proof}

We note that this result is expected as any qPUF (same as a classical PUF), can in principle, be simulated with enough computational resources. That is why the reasonable and achievable security model is usually against a qPUF in hands of the adversary for a limited time or limited query such as QPT adversaries. It is also worth mentioning that from an engineering point of view, limiting the adversary to a certain number of queries on a hardware level, can depend on the construction and it might be possible in some qPUF implementations, while might not be feasible with some others. While this is an interesting problem to be considered in qPUF implementations, from a cryptanalysis point, our security analysis against a quantum adversary who is given polynomial time in the security parameter, is independent of the construction.

Exploiting the quantum emulation algorithm introduced in Section~\ref{sec:qe} we now turn to quantum existential unforgeability, and show that no UqPUF provides quantum existential unforgeability for any $\mu\neq 1$ as defined in Definition~\ref{def:Qunf}. Note that the case $\mu=1$ corresponds to the existential challenge state being orthogonal to all the queried states in the learning phase. With $\mu=1$, the adversary is prevented from taking advantage of its quantum access to the qPUF to win the game. 

\begin{theorem}\label{th:noeuf-qRMA}
\textbf{(No UqPUF provides quantum existential unforgeability)} For any $\lambda$-UqPUF, and $0\leq \mu \leq 1-\nonnegl(\lambda)$, there exits a QPT adversary $\A$ such that
\begin{equation}
    Pr[1\leftarrow \Gn{\UqPUF}{\qEx}{\mu}(\lambda, \A)] = \nonnegl(\lambda).    
\end{equation}
\end{theorem}

\begin{proof}
We show there is a QPT adversary $\A$ who wins the game $\Gn{\UqPUF}{\qEx}{\mu}(\lambda,\A)$ with non-negligible probability in $\lambda$.
The adversary $\A$ runs the learning phase of the game $\Gn{\UqPUF}{\qEx}{\mu}(\lambda, \A)$ with $\ket{\phi_1}$ and $\ket{\phi_2}$ such that $\ket{\phi_1}$ can be any quantum state in $\HilD$ and
\begin{widetext}
\begin{equation}
    \ket{\phi_2}=\begin{cases}
\frac{1}{\sqrt{2}}(\ket{\phi_1} + \ket{\phi_3}) & if~0 \leq \mu \leq \frac{1}{2}\\
\sqrt{\mu}\ket{\phi_1} + \sqrt{1-\mu}\ket{\phi_3} & if~ \frac{1}{2} < \mu \leq 1-\nonnegl(\lambda)
\end{cases}
\end{equation}
\end{widetext}
Without loss of the generality, we assume $\A$ chooses one of the computational basis of $\HilD$ as $\ket{\phi_1}$. Then, $\A$ chooses an orthogonal state to $\ket{\phi_1}$ as $\ket{\phi_3}$ and sets $\ket{\phi_2}$ the superposition of these two states.
In the existential challenge phase, $\A$ sets $\ket{\phi_3}$ as his chosen challenge. Note that $\ket{\phi_3}$ satisfies the $\mu$-distinguishability of the challenge state with both $\ket{\phi_1}$ and $\ket{\phi_2}$. In the guess phase, to estimate the output of UqPUF to $\ket{\phi_3}$, the adversary $\A$ runs the quantum emulation (QE) algorithm defined in Section \ref{sec:qe} with the reference state $\ket{\phi_r}=\ket{\phi_2}$. 

Relying on Theorem~\ref{th:qe-fins}, the output state of Stage 1 of the QE algorithm is:
\begin{equation}
\begin{split}
    \ket{\chi_f} &= \mbraket{\phi_2}{\phi_3} \ket{\phi_2}\ket{0} + \ket{\phi_3}\ket{1} - \mbraket{\phi_2}{\phi_3}\ket{\phi_2}\ket{1}\\
    & -2\mbraket{\phi_1}{\phi_3}\ket{\phi_1}\ket{1} +2\mbraket{\phi_2}{\phi_3}\mbraket{\phi_2}{\phi_1}\ket{\phi_1}\ket{1}.
\end{split}
\end{equation}
Note that $\bra{\phi_1}\phi_3\rangle = 0$ and we set $\bra{\phi_2}\phi_3\rangle = \alpha$ and $\bra{\phi_2}\phi_1\rangle = \beta$ based on the choice of $\ket{\phi_2}$, the above equation can be simplified as:
\begin{equation}
    \ket{\chi_f} = \alpha \ket{\phi_2}\ket{0} + \ket{\phi_3}\ket{1} - \alpha \ket{\phi_2}\ket{1} +2\alpha\beta \ket{\phi_1}\ket{1}.
\end{equation}
Now, according to Theorem~\ref{th:qe-fidel}, the final fidelity in terms of the success probability of Stage 1 can be obtained by calculating the density matrix of $\ket{\chi_f}$ and tracing out the ancillas:
\begin{equation}
\begin{split}
        P_{succ-stage1} & = |\bra{\phi_2}Tr_{anc}(\ket{\chi_f}\bra{\chi_f})\ket{\phi_2}|^2 \\
        & = |\alpha^2(1+4\alpha^2\beta^2)|^2.
\end{split}
\end{equation}
We have different choices for the reference state depending on the distinguishability parameter $\mu$. For cases where the adversary is allowed to produce a new state with at least overlap half with all the states in the learning phase, by choosing the uniform superposition of the states where $\alpha = \beta = \frac{1}{\sqrt{2}}$, the output fidelity will be:
\begin{equation}
    F(\ket{\phi_3^{out'}}\bra{\phi_3^{out'}}, \ket{\phi_3^{out}}\bra{\phi_3^{out}}) \geq \sqrt{P_{succ-stage1}} = 1.
\end{equation}
where $\ket{\phi_3^{out'}}$ and $\ket{\phi_3^{out}}$ are the output of the QE algorithm and UqPUF to $\ket{\phi_3}$, respectively. 

As can be seen, these two states are completely indistinguishable So, the success probability of $\A$ for any test according to Definition~\ref{def:test} is:
\begin{equation}
        Pr[1\leftarrow \Gn{\UqPUF}{\qEx}{\mu}(\lambda, \A)] = Pr[1\leftarrow\T(\ket{\psi^{out}}^{\otimes \kappa_1}, \ket{\omega}^{\otimes \kappa_2})] = 1
\end{equation}
which is the optimal choice of the reference. On the other hand, for the cases where the adversary is restricted to produce a challenge more than half distinguishable, we can still create a superposed state with $\alpha = \sqrt{1-\mu}$ and $\beta = \sqrt{\mu}$ and end up with the following fidelity of the emulation by setting $\mu=1-\nonnegl(\lambda)$
\begin{equation}
\begin{split}
    F(\ket{\phi_3^{out'}}\bra{\phi_3^{out'}}, \ket{\phi_3^{out}}\bra{\phi_3^{out}}) & \geq |\alpha^2(1+4\alpha^2\beta^2)| \\ &= |(1-\mu)(1 + 4\mu(1-\mu))| \\&= \nonnegl(\lambda).
\end{split}
\end{equation}
Recall that the security parameter $\lambda$ includes the number of copies used in the test algorithm ($\kappa_1$, $\kappa_2$), by increasing them the probability of accepting will converge to the above fidelity thus for any $\frac{1}{2} < \mu \leq 1 -\nonnegl(\lambda)$:
\begin{widetext}
\begin{equation}
        Pr[1\leftarrow \Gn{\UqPUF}{\qEx}{\mu}(\lambda, \A)] = Pr[1\leftarrow\T(\ket{\phi_3^{out}}^{\otimes \kappa_1}, \ket{\phi_3^{out'}}^{\otimes \kappa_2})] = \nonnegl(\lambda)
\end{equation}
\end{widetext}
And the proof is complete.
\end{proof}

This theorem implies that the adversary can always generate the correct response to his chosen challenge provided that he can query it in superposition with other quantum states during the learning phase in terms of the parameter $\mu$. Note that since output quantum states in the learning phase are unknown to the adversary, the more straightforward strategy of superposing the learnt output quantum states cannot be efficiently performed. More precisely, the adversary cannot prepare the precise target superposition of the output states that are completely unknown~\cite{oszmaniec2016creating,doosti2017universal}. Hence the proposed attack is general but non-trivial.

We now further relax the level of security and consider quantum selective unforgeability. We show that any UqPUF can provide this weaker notion of security. Note that in most of the PUF-based applications such as PUF-based identification protocols, selective unforgeability is sufficient. 

We need the following lemma to prove the quantum selective unforgeability feature of UqPUFs. The lemma implies the average probability of any state in $\HilD$ to be projected in a subspace $\Hild$ where $d\leq D$. Based on this lemma, we calculate the probability of a state chosen uniformly at random from $\HilD$ to be projected in the orthogonal subspace of the adversary's database where the quantum emulation or similar attacks does not work. 

\begin{lemma}\label{lemma2}
    Let $\HilD$ be a $D$-dimensional Hilbert space and $\Hild$ a subspace of $\HilD$ with dimension $d$. Also, let $\Pi_d$ be a projector for any quantum state in $\HilD$ into $\Hild$. The average probability that any state, chosen uniformly at random from $\HilD$, $\ket\psi\underset{R}{\in}\HilD$ to be projected into $\Hild$ is equal to $\frac{d}{D}$
    \begin{equation}
        \underset{{\ket{\psi},\Pi_d}}{Pr}[|\bra{\psi}\Pi_d\ket{\psi}| = 1] = \frac{d}{D}
    \end{equation}
\end{lemma}
\begin{proof}
    The proof is mainly based on the symmetry of the Hilbert space and the fact that the probability of falling into each subspace is equal for any state uniformly picked at random.
    
    Note that Any state $\ket{\psi} \in \HilD$ can be written in terms of the orthonormal bases of $\HilD$ denoted by $\ket{b_i}$, as follows:
\begin{equation}
    \ket{\psi} = \sum_{i=0}^{D-1} \alpha_i\ket{b_i} \quad \text{with} \quad \sum_{i=0}^{D-1} |\alpha_i|^2 = 1
\end{equation}
where $\alpha_i$ are complex coefficients. A projection into a smaller subspace consists of choosing $d$ bases of $\HilD$ in the form of $\sum_{j=0}^{d-1} \ket{b_j}\bra{b_j} $. Without loss of generality, we can assume $D = md$ where $m$ is an integer. This assumption is always correct for qubit spaces. This means that the larger Hilbert space can be divided into $m$ smaller subspaces each with dimension $d$. Let $\{\ket{e_i}\}_{i=0}^{d-1}$ be a subset of $\HilD$ which makes a complete set of bases for one of the $d$-dimensional subspaces. A projector projects $\ket{\psi}$ into one of the subspaces. As $\ket{\psi}$ has been picked at random and the subspaces are symmetric, the probability of falling into each subspace is the same and equal to $\frac{1}{m}$ which is $\frac{d}{D}$. Otherwise either the sum of all probabilities would not be 1 or the $\ket\psi$ has not been picked uniformly at random from $\HilD$. This shows that on average the probability of projecting a state $\psi$ is $\frac{d}{D}$. This can also be seen by the fact that the sum of all projectors in a complete set of projectors is equal to one. In this case, we have
\begin{equation}
    \sum_{i=0}^{D-1} \Pi_i = \mathbb{I}
\end{equation}
By sandwiching $\ket{\psi}$ on both sides we have:
\begin{equation}
    \sum_{i=0}^{D-1} \bra{\psi}\Pi_i\ket{\psi} = 1.
\end{equation}
Each $\bra{\psi}\Pi_i\ket{\psi}$ is itself equal to $\sum_{j=0}^{d-1} |\mbraket{\psi}{d_{ij}}|^2$ where $\ket{d_{ij}}$s are the bases associated to the subspace that the projector $\Pi_i$ projects into. This corresponds to all the permutations of $d$ number of the coefficient $|\alpha_i|^2$ which will be $\frac{1}{d}$ on average. Since we have $\sum_{i=0}^{D-1} \frac{Pr_{\Pi_i}}{d} = 1$, we can conclude that the average probability $Pr_{\Pi}$ for all the projectors will be $\frac{d}{D}$ and the proof is complete.
\end{proof}

To establish our possibility result, we first present a preliminary theorem which demonstrates the security of the UqPUF considering an ideal test algorithm which asymptotically satisfies the notion of distance as defined in Definition~\ref{def:test-delta}.

\begin{theorem}\label{th:sel-qCM-fid}
For any unitary qPUF characterised by $\UqPUF=(\qPUFGen,\qPUFEval,\Td)$, and any non-zero $\delta$, the success probability of any QPT adversary $\A$ in the game $\Gnn{\UqPUF}{\qSel}(\lambda, \A)$ is bounded as follows:
\begin{equation}
    Pr[1\leftarrow \Gnn{\UqPUF}{\qSel}(\lambda, \A)] \leq \frac{d+1}{D}
\end{equation}
where $D$ is the dimension of the Hilbert space that the challenge quantum state is picked from, and $0\leq d \leq D-1$ is the dimension of the largest subspace of $\HilD$ that the adversary can span in the learning phase of $\Gnn{\UqPUF}{\qSel}(\lambda, \A)$.
\end{theorem}
\begin{proofsketch}
The complete proof can be found in Appendix~\ref{ap:proof-sel-fid}, here we only sketch the main idea.
We are interested in the average success probability of the adversary running the game $\Gnn{\UqPUF}{\qSel}(\lambda, \A)$. Let the subspace spanned by the learnt queries be a $d$-dimensional subspace of $\HilD$ denoted by $\Hild$. We calculate the average fidelity of the adversary's estimated output state $\ket\omega$ and the correct output $\kpo$, over all choices of the $\qSel$ challenge state $\ket{\psi}$. We require this fidelity to be greater than a value $\delta$ imposed by the $\Td$: 
\begin{equation}
    Pr[1\leftarrow \Gnn{\UqPUF}{\qSel}(\lambda, \A)] = \underset{\ket{\psi}\in\HilD}{Pr}[F(\ket\omega,\kpo) \geq \delta].    
\end{equation}
Note that because of the quantum nature of queries in the learning phase and the limited number of queries that the QPT adversary $\A$ can make, $\A$ might not have the classical description of the responses to his queries. So, we let $\A'$ be another QPT adversary who has full knowledge of $\Hild$. It is obvious that the success probability of $\A'$ would be higher than the success probability of $\A$ due to the extra knowledge that $A'$ has. So, we have
\begin{equation}
    Pr[1\leftarrow \Gnn{\UqPUF}{\qSel}(\lambda, \A)]\leq Pr[1\leftarrow \Gnn{\UqPUF}{\qSel}(\lambda, \A')]    
\end{equation}
In rest of the proof, We calculate the success probability of $\A'$ which is the higher bound for the success probability of $\A$. We write this probability in terms of its partial probabilities for the states orthogonal to $\Hild$ and the rest of the space:
\begin{widetext}
\begin{equation}
        Pr[1\leftarrow \Gnn{\UqPUF}{\qSel}(\lambda, \A')] =  \underset{\ket{\psi}\in\Hildperp}{Pr}[F \geq \delta]Pr[\ket{\psi} \in \Hildperp] + \underset{\ket{\psi}\not\in\Hildperp}{Pr}[F \geq \delta]Pr[\ket{\psi} \not\in \Hildperp].
\end{equation}
\end{widetext}
The probability of projection into the orthogonal subspace and the conjugate subspace can be obtained by calling Lemma~\ref{lemma2}:
\begin{equation}
    Pr[\ket{\psi} \in \Hildperp] = \frac{d^{\perp}}{D}
\end{equation}
where $d^{\perp}=D-d$; And
\begin{equation}
Pr[\ket{\psi} \not\in \Hildperp] = 1 - Pr[\ket{\psi} \in \Hildperp] = \frac{d}{D}
\end{equation}

We also assume there exists a QPT algorithm that its average probability over all the states not in the orthogonal subspace to estimate their outputs with $F\geq \delta$ is 1, i.e. $\underset{\ket{\psi}\not\in\Hildperp}{Pr}[F \geq \delta]=1$.

Thus, the only remaining term to calculate is the probability that the average fidelity be greater than $\delta$ in the orthogonal subspace, i.e. $\underset{\ket{\psi}\in\Hildperp}{Pr}[F \geq \delta]$. We show in Appendix~\ref{ap:proof-sel-fid} that since the $\qSel$ challenge is chosen uniformly at random from $\HilD$, the best attack strategy to achieve the desired fidelity is choosing the output state uniformly at random from $\HilD$. 

Then, we calculate the average fidelity according to Haar measure and show the average probability for non-zero fidelity is bounded by:
\begin{equation}
 \underset{\kpo\in\Hildperpo}{Pr}[F \neq 0] \leq \frac{1}{D-d}
\end{equation}
So, for non-zero $\delta$ we also have,
\begin{equation}
\underset{\kpo\in\Hildperpo}{Pr}[F \geq \delta] \leq \frac{1}{D-d}
\end{equation}
As a result, the success probability of $\A$ is bounded by 
\begin{equation}
      Pr[1\leftarrow \Gnn{\UqPUF}{\qSel}(\lambda, \A)] \leq Pr[1\leftarrow \Gnn{\UqPUF}{\qSel}(\lambda, \A')] \leq \frac{d+1}{D}
\end{equation}
And the theorem is proved.
\end{proofsketch}

\begin{theorem}\label{th:sel-qCM}
\textbf{(Any UqPUF  provides quantum selective unforgeability)} Let the test algorithm $\T$ be defined according to Definition~\ref{def:test} and satisfy the condition $Err(\kappa_1, \kappa_2) = \negl(\kappa_1, \kappa_2)$. Then, for any $\UqPUF=(\qPUFGen,\qPUFEval, \T)$ and any QPT adversary, we have:
\begin{equation}
    Pr[1\leftarrow \Gnn{\UqPUF}{\qSel}(\lambda, \A)] = \negl(\lambda).    
\end{equation}
\end{theorem}

\begin{proof}
Let $\ket\psi$ be quantum state chosen by the challenger in the selective challenge phase. Also, let $\kpo$ and $\ket\omega$ be the output of the UqPUF and the adversary $\A$ to $\ket\psi$, respectively. Note that the success probability of $\A$ in game $\Gnn{\UqPUF}{qSel}(\lambda,\A)$ is equal to the probability of the test algorithm in outputting 1:
\begin{equation}
    Pr[1\leftarrow \Gnn{\UqPUF}{\qSel}(\lambda, \A)] = Pr[1\leftarrow \T(\kpo^{\otimes \kappa_1}, \ket{\omega}^{\otimes \kappa_2})]    
\end{equation}
We denote $Pr[1 \leftarrow \T(\ket\omega^{\otimes \kappa_1}, \kpo^{\otimes \kappa_2})]$ with $Pr[1 \leftarrow \T]$ for simplicity. To calculate this probability, we consider two independent cases where leads the $\T$ outputs 1. We let $\delta$ be the threshold for $F(\ket\omega,\kpo)$ that helps us to write the $Pr[1 \leftarrow \T]$ as sum of two terms, i.e. the probability of $\T$ outputting 1 while $F\ge \delta$ and the probability of $\T$ outputting 1 while $F < \delta$:
\begin{equation}
Pr[1 \leftarrow \T] = Pr[1 \leftarrow \T , F\geq \delta] + Pr[1 \leftarrow \T , F < \delta]
\end{equation}
Let $\delta = \negl(\lambda)$ hence we have
\begin{equation}
\begin{split}
    Pr[1 \leftarrow \T] & = Pr[1 \leftarrow \T | F \geq \negl(\lambda)] Pr[F \geq \negl(\lambda)] \\ & + Pr[1 \leftarrow \T | F < \negl(\lambda)] Pr[F < \negl(\lambda)]
\end{split}
\end{equation}
and then from Theorem \ref{th:sel-qCM-fid}, it can be concluded that 
\begin{equation}
Pr[F \geq \negl(\lambda)] \leq \frac{d+1}{D}
\end{equation}
where $d$ is the dimension of the subspace spanned by the learnt queries and $D$ is the dimension of the Hilbert space that the UqPUF is defined over it. Thus, $D=2^n$ where $n$ is the number of qubits in each input/output state. Since the adversary is a QPT adversary, the number of learnt queries and as a result the value of $d$ should be polynomial in $n$, i.e. $d=poly(n)$.

Also, according to Definition \ref{def:test}, we have,
\begin{equation}
Pr[1 \leftarrow \T | F < \negl(\lambda)] = Err(\kappa_1, \kappa_2)
\end{equation}
And,
\begin{equation}
Pr[1 \leftarrow \T | F \geq \negl(\lambda)]\leq F
\end{equation}
Considering the equality cases and due to the fact that $Pr[F < \negl(\lambda)] = 1 - Pr[F \geq \negl(\lambda)]$,
\begin{equation}
Pr[1 \leftarrow \T] = Err(\kappa_1, \kappa_2) (1-\frac{d+1}{D}) + \negl(\lambda) \frac{d+1}{D}
\end{equation}
Recall that $Err(\kappa_1, \kappa_2)=\negl(\kappa_1,\kappa_2)$, $d=poly(n)$ and $D=2^n$ and hence $\frac{d+1}{D}=\negl(n)$ and the probability that the test algorithm outputs 1 is computed as
\begin{equation}
\begin{split}
Pr[1 \leftarrow \T] & = \negl(\kappa_1,\kappa_2) (1-\negl(n)) + \negl(\lambda) \negl(n) \\ &= \negl(\kappa_1,\kappa_2) + \negl(\lambda) \negl(n) 
\end{split}    
\end{equation}
Let $\lambda=f(\kappa_1,\kappa_2,n)$, therefore we have 
\begin{equation}
Pr[1\leftarrow \Gnn{\UqPUF}{\qSel}(\lambda, \A)] = Pr[1 \leftarrow \T] = \negl(\lambda)
\end{equation}
and the proof is complete.
\end{proof}

\section{Discussion and Future works}
In this section, we briefly discuss the relationship between our proposal and other types of PUFs, as well as the open questions and direction for future works. 

Here, we briefly discuss how requirements and security properties defined for cPUFs and QR-PUFs \cite{vskoric2010quantum,vskoric2012quantum} in the literature differ from or relate to what we have defined as qPUF in this paper while leaving a concrete comparison between various PUF instances for future studies.

Most of the available PUF structures use digital encoding as their inputs and outputs so that they can easily be integrated with other functionalities in Integrated Circuits (ICs). This means their input-output pairs are bit-strings. As we can encode the bit strings in computational bases of the Hilbert space, the cPUFs can be considered as special types of Unitary qPUFs (UqPUFs) that can only operate on the computational bases, i.e. map the computational bases in their input domain to other computational bases in their output range. So, our result stating that no UqPUF provides quantum existential unforgeability also shows no cPUF, assuming that they can be queried by quantum states, can provide this security notion for $\mu\neq 1$.

According to \cite{armknecht2016towards}, if a cPUF provides the min-entropy requirement (which imposes that the cPUF responses are linearly independent) then it can provide existential unforgeability \cite{armknecht2016towards} against classical adversaries with no quantum access to the cPUF. However, this requirement cannot be satisfied with most of the common cPUF structures as shown in \cite{ganji2016strong,ruhrmair2014puf,ruhrmair2010modeling,khalafalla2019pufs}. Instead of the min-entropy requirement that seems hard or impossible to be achieved, we only consider the basic assumption on PUFs that let the behaviour of PUF be unknown to anyone \cite{ruhrmair2014pufs}; and instead of existential unforgeability property which seems impossible to be achieved for both cPUFs and qPUFs, we consider the selective unforgeability property which is a weaker, yet more relevant, notion than the existential one.
 
To the best of our knowledge, there is no study on quantum security of cPUFs in the literature. We emphasise given the speedy progress in quantum technology the investigation of the security of cPUFs against quantum adversaries is crucial. The security of silicon cPUFs and the other types of cPUFs that cannot be queried by quantum states can be explored in the \emph{post-quantum (or standard) security model} where the quantum adversary has only classical interaction with the primitive while he has been equipped with a powerful quantum computer. However, for the other types of cPUF structures like optical PUFs that can naturally be queried with quantum states, the security of cPUFs need to be analysed in the quantum security model where the adversary in addition to having a quantum computer can have quantum access to the cPUF oracle. Note that quantum selective unforgeability of this type of cPUF structures can be investigated in the aforementioned model. We leave exploring these open questions for future studies. 

Another main category of PUFs that can be represented via unitary transformations, is Quantum Read-out PUFs (QR-PUFs). The original definition of QR-PUFs considered cPUFs with quantumly-encoded challenge-response pairs. \cite{vskoric2010quantum,vskoric2012quantum}. The security of QR-PUF-based identification protocols has been investigated in specific security models, such as prepare-and-resend adversaries in \cite{vskoric2010quantum,vskoric2012quantum,nikolopoulos2017continuous,goorden2014quantum,vskoric2013security,nikolopoulos2018continuous,fladung2019intercept} where either the full unitary transformation or equivalently the classical description of QR-PUF responses for any known challenge, is assumed to be public knowledge. The security of such PUF-based protocols relies on the bounds on the ability of an adversary to estimate an unknown quantum challenge sent by the verifier. 

Although our current framework as it is, will not be directly applicable to all sorts of protocols and scenarios in which QR-PUFs are defined and used due to specific sets of assumptions and adversarial models considered in these scenarios, we believe that an extended variant of QR-PUFs can be studied as a stand-alone primitive in our proposed framework. We call this extended class, Public-Database PUFs (or PDB-PUFs) which include any PUF that can be queried with quantum (or quantumly encoded) challenges, produce quantum states as responses and are modelled by a publicly known unitary transformation or a public database equivalently. Our framework provides security notions against general and quantum adversaries in the standard game-based model. Hence we can also investigate the security of PDB-PUFs, by relaxing the unknownness condition for this class.


    
It can easily be shown that in the case of PDB-PUFs the adversary has more knowledge compared to qPUFs, so, these PUFs cannot provide quantum existential unforgeability, either. But more interestingly, using our toolkit of the quantum emulation attack, one can also show that, provided that the classical description of the unitary or the responses to be known, PDB-PUFs do not even provide quantum selective unforgeability against QPT adversaries, even if the adversary is unable to efficiently estimate the challenge quantum state. To see why let us assume the challenger to be also an efficient quantum party. Hence a QPT adversary having knowledge over the database can efficiently span a subspace, including the challenge state, hence the approximate response can be produced with high fidelity using the universal quantum emulator as has been discussed in Section~\ref{sec:qe}. We should mention that the feasibility of other quantum attacks with current technologies has been discussed in \cite{vskoric2010quantum,vskoric2012quantum,nikolopoulos2017continuous,goorden2014quantum,vskoric2013security,nikolopoulos2018continuous,fladung2019intercept}. However, it remains an interesting open question when the quantum emulator attack presented in this paper can also be demonstrated on emerging quantum devices.

Another interesting direction for future work is whether the assumptions of QR-PUFs can be matched to the current framework to be able to study their provable security against stronger quantum adversaries. It seems that if one can assume the classical description of $\U_{QR}$ to be private and the challenge state can be chosen uniformly at random from the whole Hilbert space, the QR-PUFs like qPUFs can provide the quantum selective unforgeability. Although this remains an interesting open problem.



An important complementary question that we left open is the design of concrete qPUF construction based on the formal framework proposed in this work. Introducing a proper construction for quantum PUF would be much more complicated than their classical counterparts as one needs to deal with many complications of the quantum world such as decoherence. Although similar to the case of classical PUF, optical devices still remain good candidates for qPUFs and worth a formal study that would be able to show whether they satisfy all the requirements and properties of a secure qPUF. Moreover, some randomised circuit-based construction such as t-design can also be a suitable candidate for qPUF as we have recently explored \cite{kumar2021efficient}. Another challenge in the way of industrialising of the qPUFs is the need for quantum memory for some of the qPUF-based protocols. It is an interesting question that how much this resource can be reduced or even removed in different protocols. Finally, the current definition allows the study of unitary qPUFs while as also mentioned in the paper, by relaxing some of the requirements the framework could also allow for non-unitary qPUF which is another natural open question for the future studies.

\bibliographystyle{plain}
\bibliography{qPUF-qj-final}

\onecolumn\newpage
\appendix

\section{Background on Classical Physical Unclonable Functions} \label{app:classic-puf}
In this section, we briefly present the formal definition of Physical Unclonable Functions (PUFs) as found in the classical literature \cite{armknecht2016towards,ruhrmair2014pufs,brzuska2011physically}. Let a $\mathcal{D}$-family be a set of physical devices generated through the same manufacturing process. Due to unavoidable variations during manufacturing, each device has some unique features that are not easily clonable. A Physical Unclonable Function (PUF) is an operation making these features observable and measurable by the holder of the device. 

As in \cite{armknecht2016towards,brzuska2011physically}, we formalize the manufacturing process of a PUF by defining the $\mathrm{Gen}$ algorithm that takes the security parameter $\lambda$ as input and generates a PUF with an identifier $\mathbf{id}$. Note that each time the $\mathrm{Gen}$ algorithm is run, a new PUF with new $\mathbf{id}$ is built. So, we have:
\begin{equation}
    \mathrm{PUF}_{\mathbf{id}} \leftarrow \mathrm{Gen}(\lambda).
\end{equation}
Also, we define the $\mathrm{Eval}$ algorithm that takes a challenge $x$ and $\mathrm{PUF}_{\mathbf{id}}$ as inputs and generates the corresponding response $y_\mathbf{id}$ as output:
\begin{equation}
    y_{\mathbf{id}} \leftarrow \mathrm{Eval}(\mathrm{PUF}_{\mathbf{id}},x).
\end{equation}
Due to variations in the environmental conditions, for any given $\mathrm{PUF}_{\mathbf{id}}$, the $\mathrm{Eval}$ algorithm may generate a different response to the same challenge $x$. It is required that this noise be bounded as follows; if $\mathrm{Eval}(\mathrm{PUF}_{\mathbf{id}},x)$ is run several times, the maximum distance between the corresponding responses should at most be $\delta_r$. This requirement is termed the \emph{robustness requirement}.

Consider a family of $\mathrm{PUF}$ generated by the same $\mathrm{Gen}$ algorithm, and assume the algorithm $\mathrm{Eval}$ is run on all of them with a single challenge $x$. To be able to distinguish each $\mathrm{PUF}_\mathbf{id}$, it is required that the minimum distance between the corresponding responses be at least $\delta_u$. This requirement is termed the \emph{uniqueness requirement}.

The other requirement considered in \cite{armknecht2016towards} is \emph{collision-resistance}. This imposes that whenever the $\mathrm{Eval}$ algorithm is run on $\mathrm{PUF}_\mathbf{id}$ with different challenges, the minimum distance between the different responses must be at least $\delta_c$.
The parameters $\delta_r$, $\delta_u$, $\delta_c$ are determined by the security parameter $\lambda$. Robustness, uniqueness and collision-resistance are crucial for correctness of cryptographic schemes built on top of PUFs. The conditions $\delta_r \le \delta_u$ and $\delta_r \le \delta_c$ must be satisfied to allow for distinguishing different challenges and PUFs~\cite{armknecht2016towards}.

According to the above, a $(\lambda,\delta_r,\delta_u,\delta_c)$-PUF is defined as a pair of algorithms: $\mathrm{Gen}$ and $\mathrm{Eval}$ that provides the robustness, uniqueness and collision-resistance requirements. We call a $(\lambda,\delta_r,\delta_u,\delta_c)$-PUF a Classical PUF (cPUF), if the $\mathrm{Eval}$ algorithm runs on classical information such as bit strings.
Any classical function $f:\{0,1\}^n \rightarrow \{0,1\}^m$, including a cPUF's $\mathrm{Eval}$, can be modelled as a unitary transformation as follows
\begin{equation}
    \forall x\in \{0,1\}^n, \forall y\in \{0,1\}^m: U_f\ket{x,y}:=\ket{x,f(x)\oplus y}
\end{equation}
and thus a quantum adversary can query $U_f$ on any desired quantum states such as the superposition of all the classical inputs.

\section{Proof of Theorem~\ref{th:qe-fins}: Quantum Emulation Output}\label{app:qe-final-state-proof}
Here we give the full proof of Theorem~\ref{th:qe-fins} as follows.

\begin{proof}
We prove the theorem by induction. For the first block ($K=1$), according to equation~(\ref{eq:qe-recur}) and letting $\ket{\chi_0} = \ket{\psi}$ we have:
\begin{equation}
    \ket{\chi_1}= \frac{1}{2}[(I - R(\phi_r))\ket{\psi}\ket{0} + R(\phi_i)(\mathbb{I} + R(\phi_r))\ket{\psi}\ket{1}]    
\end{equation}
where the term $I - R(\phi_r) = 2\ket{\phi_r}\bra{\phi_r}$ projects the previous state to $\ket{\phi_r}$ with the coefficient $\mbraket{\phi_r}{\psi}$ and the term $R(\phi_i)(I + R(\phi_r))$ is equal to:
\begin{equation}\label{eq:qe2term}
    R(\phi_i)(I + R(\phi_r)) = 2[I - \ket{\phi_r}\bra{\phi_r} - 2\ket{\phi_i}\bra{\phi_i} + 2\bra{\phi_i}\phi_r\rangle \ket{\phi_i}\bra{\phi_r}].
\end{equation}
Thus, the final relation between all the parameters in the first block is as follows.
\begin{equation}
    \ket{\chi_1} = \mbraket{\phi_r}{\psi} \ket{\phi_r}\ket{0} + \ket{\psi}\ket{1} - \mbraket{\phi_r}{\psi} \ket{\phi_r}\ket{1} -2\mbraket{\phi_1}{\psi} \ket{\phi_1}\ket{1}
    +2\mbraket{\phi_r}{\psi}\mbraket{\phi_r}{\phi_1} \ket{\phi_1}\ket{1}
\end{equation}

As can be seen, it satisfies the form of equation~(\ref{eq:qe-fins}) where the first sum is zero and in the second sum $g_{10}=-1, g_{11}=+1$, $l'_{10} = l'_{11} = 1$, $x'_{10} = z'_{10} = 0, y'_{10} = 1$, $x'_{11} = z'_{11} = 1$ and $y'_{11} = 0$.

Now we write $\ket{\chi_{K}}$ according to equation (\ref{eq:qe-recur}), assume $\ket{\chi_{K-1}}$ is written in form of equation~(\ref{eq:qe-fins}) and show $\ket{\chi_{K}}$ also satisfies this equation.
\begin{equation}
\begin{split}
    \ket{\chi_K} &= \mbraket{\phi_r}{\chi_{K-1}} \ket{\phi_r}\ket{0} + \ket{\chi_{K-1}}\ket{1} - \mbraket{\phi_r}{\chi_{K-1}} \ket{\phi_r}\ket{1} -2\mbraket{\phi_K}{\chi_{K-1}} \ket{\phi_K}\ket{1} \\
    & +2\mbraket{\phi_r}{\chi_{K-1}}\mbraket{\phi_r}{\phi_K} \ket{\phi_K}\ket{1}
\end{split}
\end{equation}
By substituting $\ket{\chi_{K-1}}$ with its equivalent based on equation~(\ref{eq:qe-fins}), we calculate each term in the above formula. Note that the coefficient in the third term is the same as the first one with a minus sign, and the ancillary state for the first term is $\ket{0}$ while for the third term is $\ket{1}$. Thus, we only show the details of the calculation for the first term:
\begin{equation}
\begin{split}
    \mbraket{\phi_r}&{\chi_{K-1}} \ket{\phi_r}\ket{0} = \\ &
    \mbraket{\phi_r}{\psi}\ket{\phi_r}\ket{0}^{\otimes K} + \mbraket{\phi_r}{\psi}\ket{\phi_r}\ket{1}^{\otimes K-1}\ket{0} - \mbraket{\phi_r}{\psi}\ket{\phi_r}\ket{1}^{\otimes K-1}\ket{0} + \\
    & + \sum^{K-1}_{i=1}\sum^{i}_{j=0} [f_{ij} 2^{l_{ij}} |\mbraket{\phi_r}{\psi}|^{x_{ij}} |\mbraket{\phi_i}{\psi}|^{y_{ij}} |\mbraket{\phi_r}{\phi_i}|^{z_{ij}}]\ket{\phi_r}\ket{q_{anc}(i,j)}\ket{0} \\
    & + \sum^{K-1}_{i=1}\sum^{i}_{j=0} [g_{ij} 2^{l'_{ij}} |\mbraket{\phi_r}{\psi}|^{x'_{ij}} |\mbraket{\phi_i}{\psi}|^{y'_{ij}} |\mbraket{\phi_r}{\phi_i}|^{z'_{ij}+1}]\ket{\phi_i}\ket{q'_{anc}(i,j)}\ket{0}.
\end{split}
\end{equation}
The second term is calculated as follows:
\begin{equation}
\begin{split}
    \ket{\chi_{K-1}}\ket{1} &= \mbraket{\phi_r}{\psi}\ket{0}^{\otimes K-1}\ket{1} + \ket{\psi}\ket{1}^{\otimes K} -\mbraket{\phi_r}{\psi}\ket{\phi_r}\ket{1}^{\otimes K} + \\
    & + \sum^{K-1}_{i=1}\sum^{i}_{j=0} [f_{ij} 2^{l_{ij}} |\mbraket{\phi_r}{\psi}|^{x_{ij}} |\mbraket{\phi_i}{\psi}|^{y_{ij}} |\mbraket{\phi_r}{\phi_i}|^{z_{ij}}]\ket{\phi_r}\ket{q_{anc}(i,j)}\ket{1} \\
    & + \sum^{K-1}_{i=1}\sum^{i}_{j=0} [g_{ij} 2^{l'_{ij}} |\mbraket{\phi_r}{\psi}|^{x'_{ij}} |\mbraket{\phi_i}{\psi}|^{y'_{ij}} |\mbraket{\phi_r}{\phi_i}|^{z'_{ij}}]\ket{\phi_i}\ket{q'_{anc}(i,j)}\ket{1}.
\end{split}
\end{equation}
The forth term $-2\mbraket{\phi_K}{\chi_{K-1}}\ket{\phi_K}\ket{1}$ has the coefficient $-2\mbraket{\phi_K}{\chi_{K-1}}$, which produces the same sigma terms while only  $l'_{i,j}, x'_{i,j}, y'_{i,j}$ and $z'_{i,j}$ are increased by one. The fifth term $2\mbraket{\phi_r}{\chi_{K-1}}\mbraket{\phi_r}{\phi_K}\ket{\phi_K}\ket{1}$ has the coefficient  $2\mbraket{\phi_r}{\chi_{K-1}}\mbraket{\phi_r}{\phi_K}$ and similarly produces the same sigma terms where $l_{i,j}$, $x_{i,j}$, $y_{i,j}$ and $z_{i,j}$ are increased by one (Note that the $\mbraket{\phi_r}{\phi_K}$ is itself one of the terms of the sigma). Finally by adding all these terms the equation (\ref{eq:qe-fins}) is obtained and the proof is complete.
\end{proof}

\section{Lemma for the Proof of Theorem~\ref{theorem:non-unitary}}\label{ap:channel-output-ditance}
We establish the following lemma that we have used in the proof of theorem~\ref{theorem:non-unitary}.
\begin{lemma}\label{lemma:epsilon-cpt-trace-distance}
Let $\E$ be a CPT map of the for $\E(\rho) = (1-\epsilon)U \rho U^{\dagger} + \epsilon \Tilde{\E}(\rho)$ where $U$ is a unitary and $\Tilde{\E}$ is a strictly contractive CPT map. Let $\rho$ and $\sigma$ be two arbitrary density matrices with trace distance $D = \mathcal{D}_{tr}(\rho, \sigma)$. Then the following inequality holds:
\begin{equation}
    \mathcal{D}_{tr}(\rho, \sigma) - \mathcal{D}_{tr}(\E(\rho), \E(\sigma)) \leq \epsilon D
\end{equation}
\end{lemma}
\begin{proof}
We note that the first part of the channel $\E$, which outputs density matrix $U \rho U^{\dagger}$ with probability $(1-\epsilon)^2$, is a unitary and preserves the distance. As a result, for a fixed value of $\epsilon$ and fixed arbitrary states $\rho$ and $\sigma$, the difference between the trace distances of the output of $\E$ and the input states increases as $\Tilde{\E}$ becomes more contractive. As the maximum contractivity of $\Tilde{\E}$ occurs when $\Tilde{\E} = \frac{I}{d}$, then the maximum difference between the output and input trace distances is satisfies for this instance of the channel. Let $\E'(\rho) = (1-\epsilon)U \rho U^{\dagger} + \epsilon \frac{I}{d}$. Then for a fixed $\epsilon$ we will have:
\begin{equation}
    \mathcal{D}_{tr}(\rho, \sigma) - \mathcal{D}_{tr}(\E(\rho), \E(\sigma)) \leq \mathcal{D}_{tr}(\rho, \sigma) - \mathcal{D}_{tr}(\E'(\rho), \E'(\sigma))
\end{equation}

Now we calculate $\mathcal{D}_{tr}(\E'(\rho), \E'(\sigma))$ using the definition of the trace distance which is $\mathcal{D}_{tr}(\rho, \sigma) = \frac{1}{2}tr(|\rho - \sigma|)$. And $|A| = \sqrt{A^{\dagger}A}$ for a positive semidefinite matrix $A$. We calculate the trace distance as:
\begin{equation}
\begin{split}
        \mathcal{D}_{tr}(\E'(\rho), \E'(\sigma)) &= \frac{1}{2}tr[|\E'(\rho), \E'(\sigma)|] =  \frac{1}{2}tr[|(1-\epsilon)U \rho U^{\dagger} + \epsilon \frac{I}{d} - (1-\epsilon)U \sigma U^{\dagger} - \epsilon \frac{I}{d}|] \\
        & = (1-\epsilon)(\frac{1}{2}tr[|U \rho U^{\dagger} - U \sigma U^{\dagger}|]) = (1-\epsilon)\mathcal{D}_{tr}(U \rho U^{\dagger},U \sigma U^{\dagger})\\
        & = (1-\epsilon)\mathcal{D}_{tr}(\rho, \sigma)\\
        & = (1-\epsilon)D
\end{split}
\end{equation}
Finally, we can relate the desired trace distance with the above value as:
\begin{equation}
       \mathcal{D}_{tr}(\rho, \sigma) - \mathcal{D}_{tr}(\E(\rho), \E(\sigma)) \leq D - (1-\epsilon)D = \epsilon D
\end{equation}
And the lemma has been proved.
\end{proof}
\section{Full Proof of Theorem~\ref{th:sel-qCM-fid}}\label{ap:proof-sel-fid}
\begin{proof}
Let $\A$ be a QPT adversary playing the game $\Gnn{\UqPUF}{\qSel}(\lambda, \A)$ where $\UqPUF$ is defined over $\HilD$. Let $\Sin$ and $\Sout$ be the input and output database of the adversary after the learning phase both with size $k_1$, respectively. Also, Let $\Hild$ be the $d$-dimensional Hilbert space spanned by elements of $\Sin$ where $d \leq k_1$ and $\Hildout$ be the Hilbert space spanned by elements of $\Sout$ with the same dimension. $\A$ receives an unknown quantum state $\ket{\psi}$ as a challenge in the $\qSel$ challenge phase and tries to output a state $\ket{\omega}$ as close as possible to $\kpo$. We are interested in calculating the average probability that the fidelity of $\A$'s output state $\ket{\omega}$ and $\kpo$ be larger or equal to $\delta$. We calculate this probability over all the possible states chosen uniformly at random from $\HilD$.
\begin{equation}
Pr[1\leftarrow\Gnn{\UqPUF}{\qSel}(\lambda, \A)] = \underset{\ket{\psi}\in\HilD}{Pr}[F(\ket\omega,\kpo)\geq\delta]
\end{equation}
We calculate this probability over all the possible states chosen uniformly at random from $\HilD$.
We will show, for any $\delta \neq 0$, the success probability of $\A$ is negligible in $\lambda$. 

According to the game definition, as the adversary selects states of the learning phase, the classical description of these states are known for him while the corresponding responses are unknown quantum states. Let $\A'$ be the adversary who also receives the classical description of the outputs, or the complete set of bases of $\Hild$ and $\Hildout$. So, he will have a complete description of the map in the subspace; and as a result $\A'$ has a greater success probability than $\A$. \begin{equation}
Pr[1\leftarrow\Gnn{\UqPUF}{\qSel}(\lambda, \A)]\leq Pr[1\leftarrow\Gnn{\UqPUF}{\qSel}(\lambda, \A')]
\end{equation}
Therefore from now on throughout the proof, we calculate the success probability of $\A'$ who has full knowledge of the subspace.

Note that the adversary cannot enhance his knowledge of the subspace by entangling its local system to the challenges of the learning phase since the reduced density matrix of the challenge/response entangled state lies in the same subspace $\Hild$ and $\Hildout$. Hereby upper-bounding the success probability of $\A$ with the success probability of $\A'$ who has the full knowledge of the subspace we have also included the entangled queries. Thus without loss of generality and to avoid complicated notations, we consider the adversary's estimated state as a pure state $\ket{\omega}$.

Now, we partition the set of all the challenges to two parts: the challenges that are completely orthogonal to $\Hild$ subspace, and the rest of the challenges that have non-zero overlap with $\Hild$. We denote the subspace of all the states orthogonal to $\Hild$ as $\Hildperp$. We calculate the success probability of $\A'$ in terms of the following partial probabilities:
\begin{equation}
    \underset{\ket{\psi}\in\Hildperp}{Pr}[F \geq \delta] \text{  and } \underset{\ket{\psi}\not\in\Hildperp}{Pr}[F \geq \delta].
\end{equation}
Because the probability of $\ket{\psi}$ being in any particular subset is independent of the adversary's learnt queries, the success probability of $\A'$ can be written as:
\begin{equation}
        Pr[1\leftarrow\Gnn{\UqPUF}{\qSel}(\lambda, \A')] = \underset{\ket{\psi}\in\Hildperp}{Pr}[F \geq \delta]\times Pr[\ket{\psi} \in \Hildperp] + \underset{\ket{\psi}\not\in\Hildperp}{Pr}[F \geq \delta]\times Pr[\ket{\psi} \not\in \Hildperp]
\end{equation}
where $Pr[\ket{\psi} \in \Hildperp] = 1-Pr[\ket{\psi} \not\in \Hildperp]$ denotes the probability of $\ket{\psi}$ that is picked uniformly at random from $\HilD$ being projected into the subspace of $\Hildperp$. From lemma~\ref{lemma2}, we know that this probability for any subspace, is equal to the ratio of the dimensions. As $\Hildperp$ is a $D-d$ dimensional subspace, $Pr[\ket{\psi} \in \Hildperp] = \frac{D-d}{D}$ and respectively $Pr[\ket{\psi} \not\in \Hildperp] = \frac{d}{D}$. Also the probability is upper-bounded by the cases that the adversary can always win the game for $\ket{\psi} \not\in \Hildperp$. So, we have,
\begin{equation}
    Pr[1\leftarrow\Gnn{\UqPUF}{\qSel}(\lambda, \A')] \leq \underset{\ket{\psi}\in\Hildperp}{Pr}[F \geq \delta]\times(\frac{D-d}{D}) + \frac{d}{D}
\end{equation}
Finally, the only term that should be calculated is $\underset{\ket{\psi}\in\Hildperp}{Pr}[F \geq \delta]$.

Note that any $\ket{\psi}\in\HilD$ can be written in any set of full bases of $\HilD$ as $\ket{\psi} = \sum^D_{i=1}c_i\ket{e_i}$. For any $\ket{\psi}\in\Hildperp$, the set of $\{\ket{e_i}\}^D_{i=1}$ can be the a union of the bases of $\Hild$, i.e. $\{\ket{e^{in}_i}\}^d_{i=1}$ and the bases of $\Hildperp$, i.e. $\{\ket{e'_i}\}^{D}_{i=d+1}$. Note that any state in $\Hildperp$ is orthogonal to all the $\ket{e^{in}_i}$s. Thus, we write as follows
\begin{equation}
\ket{\psi} = \sum^d_{i=1}c^{in}_i\ket{e^{in}_i} + \sum^D_{i=d+1}c'_i\ket{e'_i}
\end{equation}

Recall that $\ket{\psi}\in\Hildperp$, so, $\mbraket{\psi}{e^{in}_i} = 0$ and as a result $c^{in}_i = 0$. So, 
\begin{equation}
\ket{\psi} = \sum^D_{i=d+1}c'_i\ket{e'_i}
\end{equation}
Similarly for the output state $\kpo = \sum^d_{i=1}c^{out}_i\ket{e^{out}_i} + \sum^D_{i=d+1}\alpha_i\ket{b_i}$, as the unitary preserves the inner product, $c^{out}_i = \mbraket{e^{out}_i}{\po} = \bra{e^{in}_i}U^{\dagger}U\ket{\psi} = \mbraket{e^{in}_i}{\psi} = 0$, and the correct output state can be written as
\begin{equation}
\kpo = \sum^D_{i=d+1}\alpha_i\ket{b_i} 
\end{equation}
where $\{\ket{b_i}\}^{D-d}_{i=1}$ are a set of bases for $\Hildperpo$. 
\noindent The output estimated by the adversary $\A'$ can be written as
\begin{equation}
    \ket{\omega} = \sum^d_{i=1}\beta_i\ket{e^{out}_i} + \sum^D_{i=d+1}\gamma_i\ket{q_i} 
\end{equation}
where the first term represents part of the output state, that has been produced by $\A$ from the his learnt output subspace and the second term denotes the part lies in $\Hildperpo$ with the set of bases $\{\ket{q_i}\}^{D-d}_{i=1}$. Based on the above argument, the fidelity of the first part is always zero as $\mbraket{b_i}{e^{out}_i} = 0$. 

Note that the normalization condition implies $\sum^d_{i=1}|\beta_i|^2 + \sum^D_{i=d+1}|\gamma_i|^2 = 1$. Thus for any state $\ket{\omega}$ that has a non-zero overlap with the learnt outputs, the fidelity with the correct state decreases. To make the $\A'$'s strategies optimal we assume  $\sum^{D-d}_{i=1}\gamma_i\ket{q_i} \in \Hildperpo$ where the normalization condition is $\sum^{D-d}_{i=1}|\gamma_i|^2 = 1$. 

Since there are infinite choices for set of bases orthogonal to $\{\ket{e^{out}_i}\}^d_{i=1}$, there is no way to uniquely choose or obtain the rest of the bases to complete the set. Also, another input of the adversary is the state $\ket{\psi}$ which according to the game definition, is an unknown state from a uniform distribution. As a result, the choice of the $\ket{q_i}$ bases are also independent of $\ket{e'_i}$ or $\ket{b_i}$. Thus knowing a matching pair of $(\ket{q_i},\ket{b_i})$ increases the dimension of the known subspace by one that means the adversary has more information that it is assumed to have. 

So, for each new challenge, $\A'$ produces a state $\ket{\omega} = \sum^{D-d}_{i=1}\gamma_i\ket{q_i}$ with a totally independent choice of bases. Without loss of generality we can fix the bases $\ket{q_i}$ for different $\ket{\omega}$. To calculate the success probability of $\A'$, we calculate the fidelity averaging over all the possible choices of $\psi$. As the unitary transformation preserves the distance, it maps a uniform distribution of states to a uniform distribution. This leads to a uniform distribution of all the possible $\kpo$. As a result, the average probability over all possible $\ket\psi$ is equal to the average probability over all possible $\kpo$.
\begin{equation}
    \underset{\ket{\psi}\in\Hildperp}{Pr}[F \geq \delta] = \underset{\kpo\in\Hildperpo}{Pr}[F \geq \delta].
\end{equation}
Now, we show that the adversary $\A'$ also needs to output $\ket{\omega}$ according to the uniform distribution to win the game in the average case.  

Let $\A'$ output the states according to a probability distribution $\mathfrak{D}$ which is not uniform. Then, by repeating the experiment asymptotically many times, the correct response $\kpo$ covers the whole $\Hildperpo$ while $\ket{\omega}$ covers a subspace of $\Hildperpo$. This decreases the average success probability of $\A'$. So, the best strategy for $\A'$ is to generate the states $\ket{\omega}$ such that they span the whole $\Hildperpo$, i.e. generating them according to the uniform distribution. 

Based on the above argument, and the fact that all the $\ket{\omega}$s are produced independently, we show that the average fidelity over all the $\kpo$ is equivalent to average fidelity over all the $\ket{\omega}$.

There are different methods for calculating the average fidelity \cite{zyczkowski2005average}, but most commonly the average fidelity can be written as:
\begin{equation}
    \underset{\kpo\in\Hildperpo}{\int}|\mbraket{\omega}{\po_x}|^2d\mu_x 
\end{equation}
where $d\mu$ is a measure based on which the reference state has been produced and parameterized. According to our uniformity assumption, the $d\mu$ here is the Haar measure. Note that $\ket{\omega}$ can be different for any new challenge. Now we rewrite the above average with the new parameters as:
\begin{equation}
\begin{split}
    \underset{\kpo\in\Hildperpo}{\int} F(\ket\omega,\ket{\psi_x^{out}})d\mu_x  & =  \underset{\kpo\in\Hildperpo}{\int}|\mbraket{\omega}{\po_x}|^2d\mu_x \\ & = \underset{\kpo\in\Hildperpo}{\int}|\sum^{D-d}_{i=1}\overline{\gamma_i}\mbraket{q_i}{\po_x}|^2d\mu_x \\
    & = \underset{\kpo\in\Hildperpo}{\int}|\sum^{D-d}_{i=1}\overline{\gamma_{i_{x}}}\mbraket{q_i}{\po}|^2d\mu_x \\ & = \underset{\ket{\omega}\in\Hildperpo}{\int}|\mbraket{\omega_x}{\po}|^2d\mu_x \\
    & = \underset{\ket\omega\in\Hildperpo}{\int} F(\ket{\omega_x},\ket{\psi^{out}})d\mu_x
\end{split}
\end{equation}

The above equality holds since the fidelity is a symmetric function of two states and the measure of integral is the same for both cases. We use this equality for averaging all the possible outputs for one $\kpo$. Recall that we aim to calculate the probability of the average fidelity being greater than $\delta$. To this end, we first calculate a more general probability that is the probability of the average fidelity to be non-zero. As we have 
\begin{equation}
    \underset{\ket{\omega}\in\Hildperpo}{Pr}[F \neq 0] + \underset{\ket{\omega}\in\Hildperpo}{Pr}[F = 0] = 1,
\end{equation}    
we calculate the probability of the zero fidelity for simplicity. So,
\begin{equation}
\begin{split}
    \underset{\ket{\omega}\in\Hildperpo}{Pr}[F = 0] &= \underset{\ket{\omega}\in\Hildperpo}{Pr}[|\mbraket{\omega}{\po}|^2 = 0] \\ & = Pr[(\int|\sum^{D-d}_{i=1}\overline{\gamma_{i_{x}}}\mbraket{q_i}{\po}|^2d\mu_x) = 0] \\
    & = \underset{x}{Pr}[(\sum^{D-d}_{i,j=1}\overline{\gamma_{i_{x}}}\alpha_j\mbraket{q_{i_{x}}}{b_j})^2 = 0]
\end{split}
\end{equation}
Based on the Cauchy–Schwarz inequality we have the following inequality:
\begin{equation}
    [\sum^{D-d}_{i,j=1}\overline{\gamma_{i_{x}}}\alpha_j\mbraket{q_{i}}{b_j}]^2 \geq  \sum^{D-d}_{i,j=1}|\overline{\gamma_{i_{x}}}\alpha_j|^2|\mbraket{q_{i}}{b_j}|^2
\end{equation}
where, 
\begin{equation}
\begin{split}
    \sum^{D-d}_{i,j=1}|\overline{\gamma_{i_{x}}}\alpha_j|^2|\mbraket{q_{i}}{b_j}|^2 & =
     \sum^{D-d}_{i,j=1}|\overline{\gamma_{i_{x}}}\alpha_j|^2|\mbraket{q_{i}}{b_j}\mbraket{b_j}{q_{i}}| = \sum^{D-d}_{i,j=1}|\overline{\gamma_{i_{x}}}\alpha_j|^2|\bra{q_{i}}\Pi_j\ket{q_{i}}|
\end{split}
\end{equation}
So, we have, 
\begin{equation}
\underset{\ket{\omega}\in\Hildperpo}{Pr}[F = 0] \geq \underset{x}{Pr}[\sum^{D-d}_{i,j=1}|\overline{\gamma_{i_{x}}}\alpha_j|^2|\bra{q_{i}}\Pi_j\ket{q_{i}}|=0]
\end{equation}
The smaller term is the probability of $\ket\omega$ being projected into the orthogonal subspace of a space that only includes $\kpo$ averaging over all the projectors. We call again Lemma~\ref{lemma2}. As the target subspace includes only one vector of the Hilbert space, the dimension of the orthogonal subspace is always one dimension less. Recall that $d^{\perp}=D-d$, the dimension of the intended orthogonal subspace is equal to $D-d-1$. So,
\begin{equation}
\begin{split}
    & \underset{x}{Pr}[(\sum^{D-d}_{i,j=1}|\overline{\gamma_{i_{x}}}\alpha_j|^2|\bra{q_{i}}\Pi_j\ket{q_{i}}|) = 0] = \frac{D-d-1}{D-d} \Rightarrow \\
    & \underset{\ket{\omega}\in\Hildperpo}{Pr}[F = 0] \geq \frac{D-d-1}{D-d}
\end{split}
\end{equation}
And as a result, 
\begin{equation}
    \underset{\kpo\in\Hildperpo}{Pr}[|\mbraket{\omega}{\po}| \neq 0] \leq \frac{1}{D-d}
\end{equation}    
So, for any non-zero $\delta$ we have,
\begin{equation}
    \underset{\ket{\psi}\in\Hildperp}{Pr}[|\mbraket{\omega}{\po}| \geq \delta] \leq \frac{1}{D-d}
\end{equation}    
Thus, the success probability of $\A'$ is
\begin{equation}
    Pr[1\leftarrow\Gnn{\UqPUF}{\qSel}(\lambda, \A')] = \frac{1}{D-d}\times(\frac{D-d}{D}) + \frac{d}{D} = \frac{d+1}{D}
\end{equation}
And the success probability of $\A$ is bounded by $\frac{d+1}{D}$,
\begin{equation}
Pr[1\leftarrow\Gnn{\UqPUF}{\qSel}(\lambda, \A)] \leq \frac{d+1}{D}
\end{equation}
and the theorem has been proved.
\end{proof}

\end{document}